\documentclass[lettersize,journal]{IEEEtran}

\usepackage{graphicx}
\usepackage{booktabs}
\usepackage{algorithm}
\usepackage{algorithmic}
\usepackage{bbding}
\usepackage{enumerate}
\usepackage{amsthm}
\usepackage{multirow}
\usepackage{color}
\usepackage{framed}
\usepackage{verbatim}
\usepackage{thmtools} 
\usepackage{thm-restate}
\usepackage{subcaption}
\usepackage{epsfig}
\usepackage{amsmath}
\usepackage{amssymb}
\usepackage{hyperref}
\usepackage[capitalize]{cleveref}
\newcommand{\name}{BPFL~}

\usepackage[accsupp]{axessibility}  

\usepackage{orcidlink}
\usepackage{wrapfig}
\usepackage{graphicx}
\usepackage{flushend}
\newtheorem{definition}{Definition}

\begin{document}

\title{Efficient Byzantine-Robust and Provably Privacy-Preserving Federated Learning} 

\author{Chenfei Nie, 
        Qiang Li, 
        Yuxin Yang, 
        Yuede Ji, 
        and Binghui Wang,~\IEEEmembership{Member,~IEEE}
\thanks{
Chenfei Nie and Qiang Li are with the College of Computer Science and Technology, Jilin University, China. }
\thanks{
Yuxin Yang is with the College of Computer Science and Technology, Jilin University, China, and Department of Computer Science, Illinois Institute of Technology, USA.} 
\thanks{Yuede Ji is with the Department of Computer Science and Engineering at the University of Texas at Arlington, USA.}
\thanks{Binghui Wang is with the Department of Computer Science, Illinois Institute of Technology, USA. Wang is the corresponding author (bwang70@iit.edu).}
}

\markboth{IEEE TRANSACTIONS ON DEPENDABLE AND SECURE COMPUTING}
{Shell \MakeLowercase{\textit{et al.}}: A Sample Article Using IEEEtran.cls for IEEE Journals}

\maketitle

\begin{abstract}
Federated learning (FL) is an emerging distributed learning paradigm without sharing participating clients' private data. However, existing works show that FL is vulnerable to both Byzantine (security) attacks and data reconstruction (privacy) attacks. Almost all the existing FL defenses only address one of the two attacks. 
A few defenses address the two attacks, but they are not efficient and effective enough. 
We propose BPFL, an efficient Byzantine-robust and provably privacy-preserving FL method that addresses all the issues. 
Specifically, we draw on 
state-of-the-art Byzantine-robust FL methods and use similarity metrics to measure the robustness of each participating client  in FL. 
The validity of clients are formulated as circuit constraints on similarity metrics  and verified via a zero-knowledge proof. 
Moreover, the client models are masked by a shared random vector, which is generated based on homomorphic encryption. In doing so, the server receives the masked client models rather than the true ones, which are proven to be private. BPFL is also efficient due to the usage of non-interactive zero-knowledge proof. 
Experimental results on various datasets show that our \name is efficient, Byzantine-robust, and privacy-preserving.

\end{abstract}


\section{Introduction}

Federated learning (FL) \cite{mcmahan2017communication}, an emerging distributed learning paradigm, enables multiple clients to collaboratively train a model coordinated by 
a server, where the client data are kept locally and do not share with any other clients/server
during the entire course of learning. 
As data from each participating client do not need to be shared, FL provides a baseline level of privacy.  
Many applications, e.g., 
Google's Gboard \cite{DBLP:journals/corr/abs-1902-01046}, healthcare informatics \cite{xu2021federated},  and credit risk prediction \cite{yang2019federated}, have shown the potential of FL in the real-world.
However, recent works have shown that the current FL design faces both security and privacy issues. 

{On one hand, FL is vulnerable to 
Byzantine attacks
(also called model poisoning attacks) 
\cite{bhagoji2019analyzing,fang2020local, baruch2019little, shejwalkar2021manipulating,bagdasaryan2020backdoor,xie2019dba},
where a few malicious clients can significantly reduce the overall performance by injecting carefully designed poisoned local models during FL training.} 
On the other hand, many works
~\cite{hitaj2017deep,zhu2019deep,wang2019beyond,geiping2020inverting,yin2021see,jeon2021gradient,luo2021feature,balunovic2021bayesian,fowl2021robbing,sun2020provable} show that FL is 
vulnerable to privacy attacks~\cite{NoorbakhshZHW24,arevalo2024task,feng2024universally,sun2021soteria}, particularly 
data reconstruction 
attacks 
where an honest-but-curious server can successfully recover the private client data from the shared client models.

\begin{figure}[!t]
    \centering
    \includegraphics[width=8cm]{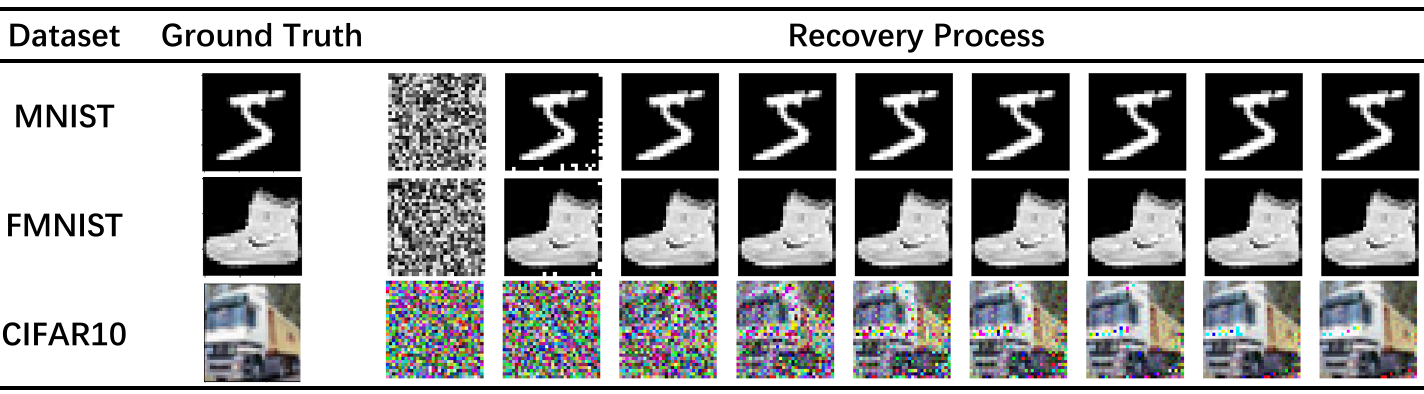}
    \caption{An honest-but-curious server 
    recovers the raw data from the shared client models trained by FLTrust.}
    \label{fig:1}
    \vspace{-2mm}
\end{figure}

To address the security issue, several Byzantine-robust FL methods have been proposed~\cite{blanchard2017machine,chen2017distributed,yin2018byzantine,guerraoui2018hidden,chen2018draco,pillutla2019robust,xie2019zeno,wu2020federated,cao2020fltrust,farhadkhani2022byzantine,karimireddy2021learning,zhang2022fldetector}. Though with different techniques, 
the main idea of these 
schemes is 
that the server performs statistical analysis on client models and 
uncovers malicious client models as those largely deviate from  others based on some similarities metrics. 
For instance, in the state-of-the-art FLTrust~\cite{cao2020fltrust} and FLDetector~\cite{zhang2022fldetector}, the server holds a clean validation dataset and uses it to train a reference model. The server then assigns a trust score to each client model based on the cosine similarity between the client model and the reference model. The trust score will be used as a weight for the local model to participate in the aggregation.
\par

All the existing Byzantine-robust FL methods are based on \emph{plaintext} (i.e., shared client model gradients or parameters). 
We note that these Byzantine-robust FL methods also face the privacy issue. 
For instance, we successfully recover the raw training data from the shared client models trained by the 
FLTrust~\cite{cao2020fltrust} using the DLG attack \cite{zhu2019deep} (See Figure \ref{fig:1}). 
To mitigate privacy attacks, various provably privacy-preserving FL mechanisms have been designed, and they are mainly based on  three techniques: 
differential privacy (DP)~\cite{pathak2010multiparty,shokri2015privacy,hamm2016learning,mcmahan2018learning,geyer2017differentially,wei2020federated}, secure multi-party computation (MPC)~\cite{danner2015fully,mohassel2017secureml,bonawitz2017practical,melis2019exploiting}, and homomorphic encryption (HE)~\cite{aono2017privacy,zhang2020batchcrypt}. 
However, there exist several challenges to prevent these privacy-preserving FL mechanisms  being Byzantine-robust. 
First, the inherent design of these mechanisms is not compatible with Byzantine-robust schemes. This is because almost all the existing privacy-preserving mechanisms only support  linear operations, while Byzantine-robust schemes are often non-linear. 

Second, these mechanisms either have high utility losses (e.g., DP-based), or induce large communication and computation overheads (e.g., MPC-based). 
More details about existing Byzantine-robust and privacy-preserving FL works are shown in Section~\ref{sec:related}.

\par
We advocate that a 
practical FL system should maintain the privacy of the participating clients' data, ensure the robustness~\cite{yang2024distributed, yang2024learning} against malicious clients  
during the entire learning, and be computation and communication efficient.
However, as discussed above, 
all the current FL works only satisfy part of these requirements.   
We aim to design a novel FL that achieves the following {goals} simultaneously:
\begin{enumerate}
  \item \textbf{Byzantine-robust}: The server can detect invalid and malicious local models submitted by clients and refuse them to participate in global model aggregation.
  \item \textbf{Privacy-preserving}: The server cannot infer clients' private  data during the entire FL training. 
  \item \textbf{Efficient}:  
  Our method should not incur too many computation and communication overheads. 
\end{enumerate}

Specifically, we propose an efficient Byzantine-robust and privacy-preserving FL method termed {\bf BPFL}. 
To ensure Byzantine-robust, we draw ideas from existing Byzantine-robust methods such as Krum~\cite{blanchard2017machine} and 
FLTrust~\cite{cao2020fltrust}, where the server in BPFL also holds a clean validation dataset to train a reference model and uses similarity metrics to measure the maliciousness/robustness of client models. 
Unlike existing methods, 
we use the Zero-Knowledge Proof (ZKP)~\cite{goldwasser1989knowledge} to verify the validity of 
client models. 
Particularly, we use both cosine similarity and Euclidean distance as circuit constraints in the ZKP to verify the validity of the client model, and the proofs are generated by the clients and verified by the server. The client model that fails to validate will be rejected to participate in the aggregation. 
To guarantee privacy-preserving, 
the client model submitted to the server will be masked by a random vector; and we design a mask vector negotiation protocol (MVNP), based on HE, 
to generate a shared mask for all clients without sharing each client's data to others. The server then receives the masked client models rather than the true ones and we prove that the server cannot infer any useful data information from the masked client models. Finally, 
due to the non-interactive properties of ZKP and the efficiency to generate circuit constraints,  \name is also  communication and computation efficient. 
In summary, our contributions are as follows:
\begin{itemize}
  \item BPFL seamlessly integrates the ideas of existing Byzantine-robust FL methods,  zero-knowledge proof, and HE into a unified framework. 
  \item Evaluations on synthetic and real-world datasets show BPFL can defend against both Byzantine attacks and data reconstruction attacks, with small computation and communication overheads.
\end{itemize}

\begin{figure*}[tbp]
    \centering
    \includegraphics[width=0.95\textwidth]{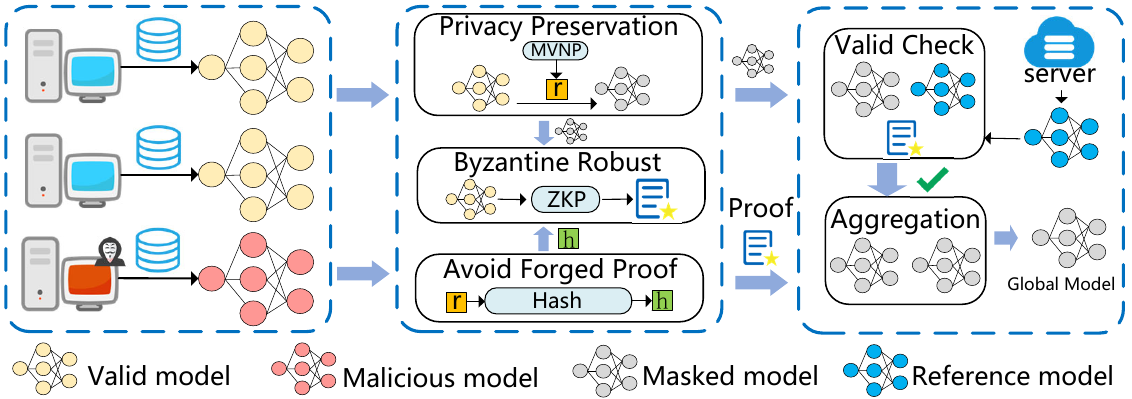}
    \caption{Overview of BPFL. 
    }
    \label{fig:overview}
\end{figure*}

\section{Preliminaries and Threat Model}

{
\noindent{\bf Federated Learning (FL).}
Suppose we are given a set of $n$ clients $C=\{{C_{1},C_{2},\cdots,C_{n}}\}$ and a server $S$, where 
each client $C_{i} \in C$ holds a dataset $D_{i}$. 
At the beginning, the server initializes a global model $w_g^1$.   In each round $t$, each client $C_{i}$ downloads the global model $w_g^{t}$ from the server $S$ and updates its local model $w_{i}^{t} $ by minimizing a task-dependent loss defined on its dataset $D_{i}$.  
Then the server $S$ collects the updated local client models $\{w_{i}^{t}\}$ and updates the global model $w_g^{t+1}$ for the next round via an aggregation algorithm. For instance, when using the most common federated averaging (FedAvg) aggregation~\cite{mcmahan2017communication}, the updated global model is: $ w_g^{t+1}\gets  \sum_{i=1}^{n} \frac{1}{n} w_{i}^{t} $, 
assuming that all clients have the same dataset size. 
}

\noindent {\bf Homomorphic encryption (HE).} Homomorphic encryption (HE) schemes allow certain mathematical operations to be performed directly on ciphertexts, without prior decryption. The Paillier encryption algorithm is an encryption algorithm for additive homomorphism and is semantically secure. The construction of the Paillier encryption system consists of three algorithms: key generation algorithm, encryption algorithm and decryption algorithm.
\begin{itemize}
 \item $(\textit{\textbf{pk}},\textit{\textbf{sk}}) \gets KeyGen(1^{\kappa})$: The $KeyGen$ algorithm outputs a public key $\textit{\textbf{pk}}$ and a private key $\textit{\textbf{sk}}$ with inputting a security parameter $\kappa$. 
	\item $c \gets Enc(\textit{\textbf{pk}},m)$: The $Enc(\cdot)$ algorithm uses the public key $\textit{\textbf{pk}}$ to encrypt plaintext $m$ to ciphertext $c$.
	\item $m \gets Dec(\textit{\textbf{sk}},c)$: The $Dec(\cdot)$ algorithm uses the private key $\textit{\textbf{sk}}$ to decrypt ciphertext $c$ into plaintext $m$.
\end{itemize}
The Paillier encryption scheme supports homomorphic addition operations 
and scalar multiplication operations. For ciphertext $c_{1}=Enc(\textit{\textbf{pk}},m_{1})$ , $ c_{2}=Enc(\textit{\textbf{pk}},m_{2})$ and a scalar $l$, it satisfies:
$Dec(\textit{\textbf{sk}},c_{1} \times c_{2})=m_{1} + m_{2}$
and $Dec(\textit{\textbf{sk}},\textrm{pow}(c_{i},l))=m_{i} \times l$, for $i={1,2}$. 

\noindent {\bf Similarity metrics in Byzantine-robust FL.}
Byzantine-robust FL methods often use similarity metrics for robust aggregation. E.g.,  Krum~\cite{blanchard2017machine} uses the  Euclidean distance, 
while  
FLTrust~\cite{cao2020fltrust} uses the cosine similarity. 
We simply introduce these two metrics, 
which are of interest in this work. 
Given two $m$-dimensional vectors ${\bf u} =\{{u_{1},u_{2},\cdots,u_{m}}\}$ and ${\bf v}=\{{v_{1},v_{2},\cdots,v_{m}}\}$, 
the Euclidean distance between ${\bf u}$ and ${\bf v}$ 
is defined as $s_{Euc}({\bf u},{\bf v}) = \sqrt{\sum\nolimits_{i=1}^{m}({u_i-v_i})^{2}}$, 
and the cosine similarity between ${\bf u}$ and ${\bf v}$ is $s_{cos}({\bf u}, {\bf v})=\frac{{\bf u} \cdot {\bf v}}{\left \| {\bf u} \right \| \left \| {\bf v}\right \|} =\frac{\sum_{i=1}^{m}u_i \cdot v_i}{\sqrt{\sum_{i=1}^{m}(u_i)^{2}}\cdot  \sqrt{\sum_{i=1}^{m}(v_i)^{2}}}$.

\noindent {\bf Zero-Knowledge Proof (ZKP).} 
In this work, we use the Groth16 scheme \cite{groth2016size}, which is a famous Zero-Knowledge Succinct Non-Interactive Arguments of Knowledge (zk-SNARK) technology. Groth16 uses the Rank-1 Constraint System (R1CS) to describe the arithmetic circuit and then converts the R1CS satisfiability problem into a Quadratic Arithmetic Program (QAP)  \cite{gennaro2013quadratic} satisfiability problem. The Groth16 algorithm contains three parts, expressed as $\Pi =(Setup,Prove,Verify)$. Given a polynomial time judgeable binary relation $\mathcal{R}$ described by R1CS, Groth16 has the following three steps
: 
\begin{itemize}
    \item $(pk,vk)\gets Setup(1^{\lambda},\mathcal{R})$: It takes the security parameter $\lambda$ and the relationship $\mathcal{R}$ as inputs. 
    The algorithm reduces the satisfiability problem of arithmetic circuits to the QAP satisfiability problem, and then a proving key $pk$ and a verification key $vk$ are generated. 
    \item $\pi \gets Prove(pk,I_p,I_a)$: The prover generates a proof $\pi$ by taking $pk, I_p, I_a$ as input, where $I_a$ is the auxiliary input that only the prover knows and $I_p$ is the primary input that both prover and verifier know. 
    \item $0/1 \gets Verify(vk,\pi,I_p)$: The verifier verifies the proof $\pi$ by taking $vk, \pi, I_p$ as input. Only if the proof passes the verification the verifier will get 1 else get 0. 
\end{itemize}

\noindent{\bf More details about Groth16.} Groth16 uses the Rank-1 Constraint System (R1CS) to describe the arithmetic circuit and then converts the R1CS satisfiability problem into a Quadratic Arithmetic Program (QAP)  \cite{gennaro2013quadratic} satisfiability problem.The R1CS is a tuple containing seven elements $(\mathbb{F}, \textbf{A},\textbf{B},\textbf{C},\textbf{io},m,n)$, where $ \mathbb{F}$ is the finite field, $\textbf{io}$ is the public input and output vectors, $\textbf{A},\textbf{B},\textbf{C} \in \mathbb{F} ^{m\times m}, m\ge \left | \textbf{io}+1 \right |, n$ is the maximum number of nonzero values in all matrices. 
R1CS is satisfiable if and only if there exists a proof $w \in \mathbb{F} ^{m-\left | \textbf{io} \right | -1} $ such that $(\textbf{A}z)\odot (\textbf{B}z)=(\textbf{C}z)$, where $z=(\textbf{io},1,w)^{T}$ and $\odot$ is the hadamard product. 
The QAP, expressed as $\mathcal{Q}=(t(z),\mathcal{U},\mathcal{W},\mathcal{Y})$, in finite field $\mathbb{F}$ 
contains three polynomials  $\mathcal{U}={u_{k}(z),\mathcal{W}={w_{k}(z)},\mathcal{Y}={y_{k}(z)}(k\in 0\cup\left[m\right])} $ and a target polynomial $t(z)$. Let $(c_{1},c_{2},\cdots ,c_{N})$ be the public input, $\mathcal{Q}$ is \emph{satisfiable} if and only if there exists $(c_{N+1},c_{N+2},\cdots ,c_{m})$ such that $t(z)$ divides $p(z)$, where
\begin{equation}
\begin{split}
p(z)=(u_{0}(z)+\sum_{k=1}^{m}c_{k} \cdot u_{k}(z))\cdot (w_{0}(z)+\sum_{k=1}^{m}c_{k}\cdot w_{k}(z))\\-(y_0(z)+\sum_{k=1}^{m}c_k \cdot y_k(z))\nonumber.
\end{split}
\end{equation}
That is, there exists a polynomial $h(z)$ such that $p(z)-h(z)t(z)=0$.
\par
The three phases in the protocol of Groth16 run as follow: 
\begin{itemize}
    \item Setup Phase.
    \begin{enumerate}
        \item Generate QAP$(t(z),\mathcal{U},\mathcal{W},\mathcal{Y})$ according to arithmetic circuit $C$.
        \item Generate $\mathbb{G}_1,\mathbb{G}_2,\mathbb{G}_T$ which are groups of prime order $p$. The pairing $e:\mathbb{G}_1\times \mathbb{G}_2 \to \mathbb{G}_T$ is a bilinear map. Let $[a]_1$ be $g^a$, $[b]_2$ be $h^b$, $[c]_T$ be $e(g,h)^c$. Select random numbers $\alpha,\beta,\gamma,\delta,s \overset{\$}{\gets}\mathbb{F} $.
        \item Generate reference string $\sigma=([\sigma_1]_1,[\sigma_2]_2)$ and an analog threshold $\tau=(\alpha,\beta,\gamma,\delta,s)$, with 
        \begin{equation}
            \begin{split}
                &\sigma _1=\binom{\alpha,\beta,\delta,\{s^i\}_{i=0}^{d-1},\{ \frac{\beta u_i(s)+\alpha w_i(s)+y_i(s)}{\gamma } \}_{i=0}^{N}}{\{\frac{\beta u_i(s)+\alpha w_i(s)+y_i(s)}{\delta  } \}_{i=N+1}^{m},\{ \frac{s^it(s)}{\delta }\}_{i=0}^{d-2}} ,\\ &\sigma _2=(\beta ,\gamma ,\delta ,\{ s^i\}_{i=0}^{d-1})\nonumber.
            \end{split}
        \end{equation}
    \end{enumerate}
    
    \item Prove Phase. The prover $\mathcal{P}$ generates the proof. $\mathcal{P}$ randomly selects $r_1,r_2\overset{\$}{\gets } \mathbb{F} $ and generates proof $\pi =([A]_1,[C]_1,[B]_2)$, where
    $$A=\alpha +\sum_{i=0}^{m}c_iu_i(s)+r_1\delta ,B=\beta +\sum_{i=0}^{m}c_iw_i(s) +r_2\delta,$$
    \begin{equation}
        \begin{split}
            C=\frac{\sum_{i=N+1}^{m}c_i(\beta u_i(s)+\alpha w_i(s)+y_i(s))+h(s)t(s) }{\delta } \\+Ar_2+Br_1-r_1r_2\delta \nonumber.
        \end{split}
    \end{equation}
    
    \item Verify Phase. Validator $\mathcal{V}$ verifies the proof. The $\mathcal{V}$ check:
    \begin{equation}
        \begin{split}
                &e([A]_1,[B]_2)\overset{?}{=}e([\alpha ]_1,[\beta ]_2)\\ &e(\sum_{i=0}^{N}c_i[\frac{\beta u_i(s)+\alpha w_i(s)+y_i(s)}{\gamma } ]_1,[\gamma ]_2)e([C]_1,[\delta ]_2)\nonumber.
        \end{split}
    \end{equation}

\end{itemize}

\noindent {\bf Threat model.}
We consider  
two types of adversaries:
\begin{itemize}
    \item \textbf{Malicious clients}. Malicious clients can 
    deviate from the protocol to corrupt the global model. Their attack goal includes: (1) submitting a poisoned local model to the server and compromising the final aggregation; (2) attempting to submit forged proofs to deceive the server; (3) failing the valid check of an honest client by providing incorrect proof parameters. Without loss of generality, 
    we assume 
    the number of malicious clients is less than  
    $50\%$ of the total clients. 
    \item \textbf{Honest-but-curious server}. The server follows the protocol, but aims to  infer clients' private information (e.g., raw data) via analyzing the received client models. We assume  the server does not collude with malicious clients. 
    Following 
    \cite{cao2020fltrust}, we also  assume the server holds a small and clean dataset $D_S$, which is used to train a reference model $w_S$ for malicious client models detection.
\end{itemize}
\par 

\begin{figure*}[tbp]
    \centering
    \begin{subfigure}{0.3\textwidth}
        \includegraphics[width=0.9\textwidth]{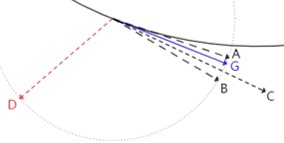}
        \caption{Unable to detect malicious update $\boldsymbol C$ only with cosine similarity}
        \label{fig:short-a}
    \end{subfigure}
    \begin{subfigure}{0.3\textwidth}
        \includegraphics[width=0.9\textwidth]{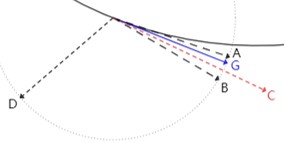}
        \caption{Unable to detect malicious update $\boldsymbol D$ when only using  Euclidean distance}
        \label{fig:short-a}
    \end{subfigure}
    \begin{subfigure}{0.3\textwidth}
        \includegraphics[width=0.9\textwidth]{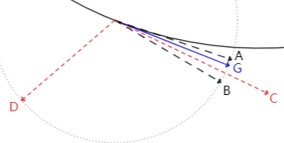}
        \caption{When using both metrics, both malicious update can be detected}
        \label{fig:short-a}
    \end{subfigure}
    \caption{Motivation of using both the cosine similarity and Euclidean distance to detect malicious model updates. $\boldsymbol G$ is the real update, $\boldsymbol A$ and $\boldsymbol B$ are valid updates, while $\boldsymbol C$ and $\boldsymbol D$ are malicious updates. 
    }
    \label{fig:fig_micPerMon}
\end{figure*}

\section{Design of \name}
\label{sec:BPFL}
We give a high-level overview of \name in Figure \ref{fig:overview}. For better understanding, we divide the structure of BPFL into three parts: Byzantine robust part, privacy preservation part, and avoiding forged proof part. In this section, we first introduce our zero-knowledge proof-inspired valid robustness check algorithm that uses the similarity metrics in the state-of-the-art robust FL methods. Next, we show how to use a random mask vector $r$ generated by homomorphic encryption to protect the client's privacy. However, for a malicious client, it may use a forged random mask vector to generate fake proofs to spoof the server, so we further leverage hash for vectors to address the possible forged proofs caused by malicious clients. Finally, we show the entire workflow of our BPFL. The important notations used in this paper are shown in 
Table \ref{tbl:notations}.

\begin{table}[!h]
\center
\caption{Important notations used in the paper.}
\addtolength{\tabcolsep}{8pt}
\label{tbl:notations}
\begin{tabular}{c|c}
\hline
{\bf Notation} & {\bf Meaning}  \\  \hline 
{\bf } $w_i$& the local model for  client  $C_i$ \\ \hline
{\bf } $w_S$& the reference model on server \\ \hline
{\bf } $w_g$& the global model \\ \hline
{\bf } $w_i^t$& $w_i$ in round  $t$  \\ \hline
{\bf } $w_S^t$& $w_S$ in round  $t$  \\ \hline
{\bf } $w_g^t$& $w_g$ in round  $t$  \\ \hline 
{\bf } $r$& random mask vector\\ \hline
{\bf } $\bar{w}_i$& the masked local model for client $C_i$ \\ \hline
{\bf } $\bar{w}_g$& the masked global model  \\

\hline
{\bf } $\tau_c$&  threshold for cosine similarity  \\ \hline
{\bf } $\tau_e$&  threshold for Euclidean distance  \\ \hline
{\bf } $\sigma _t$&  proof parameters for round  $t$  \\ \hline

{\bf } $I_a$&  the auxiliary input\\ \hline
{\bf } $I_p$&  the primary input\\ \hline
{\bf } $\pi_i^t$&  the proof of client $C_i$ for round  $t$  \\ \hline

{\bf } $\tilde{r}$& the forged $r$  \\ \hline
{\bf } $\tilde{w}$& the forged $w$  \\ \hline
{\bf } $\tilde{\pi}$& the forged $\pi$  \\ \hline
\end{tabular}
\end{table}

\subsection{Valid Robustness Check for Local Models {via Zero-Knowledge Proof}}
As shown in the recent works~\cite{blanchard2017machine,guerraoui2018hidden,cao2020fltrust}, 
a Byzantine-robust FL method should require that 
the local model is similar to the global model, in terms of both the model update direction and the model magnitude (also see an example in Figure~\ref{fig:fig_micPerMon}). 
Motivated by this, 
we use both cosine similarity and Euclidean distance to measure the validity of a client model. 
To avoid complex multi-party computation,  
we use non-interactive zero-knowledge proof (ZKP) to perform the valid check. 
Recall that the server holds a benign reference model $w_S$ (trained on its clean dataset $D_S$), which we will treat as the reference for local model comparison.  
We adopt the cosine similarity and Euclidean distance between each client model $w_i$ and the reference model $w_S$. Then, a valid robust local model should satisfy:
\begin{align}
\small
\begin{split}
\label{eq:1}
    \sqrt{\sum_{j=1}^{m}({w_{i}^j-w_{S}^j})^{2}}\le \tau_e, 
    \frac{\sum_{j=1}^{m}w_{i}^j \times w_{S}^j}{\sqrt{\sum_{j=1}^{m}(w_{i}^j)^{2}}\times \sqrt{\sum_{j=1}^{m}(w_{S}^j)^{2}}}\ge \tau_c,  
\end{split}
\end{align}%
where  $\tau_e$  and $\tau_c$  are the threshold for cosine similarity and Euclidean distance, e.g., defined by the server.  
Note that Equation \ref{eq:1} 
contains  square root and division computations that are difficult to be expressed by arithmetic circuits, as they only consist of addition and multiplication gates. 
To address it, we conduct some transformations in order to satisfy the Groth16 scheme. 
Consider that all computations in Groth16 are in a finite field, we use fixed-point numbers to approximate floating-point numbers. By default, we first define, e.g.,  $k=2^{16}$ and transform the number $x$ 
to be $x'$ by setting $x'=kx$, and then simply truncate the fractional part. In doing so, the to be verified Equation \ref{eq:1} in the proof circuit are expressed as: 
\begin{align}
    & \sum_{j=1}^{m}(k{w_{i}^j-kw_{S}^j})^{2}\le (k\tau_e)^2, \label{eq:2} 
    \\ 
    & \big(k\sum_{j=1}^{m}(kw_{i}^j \times kw_{S}^j)\big)^2 \ge (k\tau_c)^2 \times{\sum_{j=1}^{m}(kw_{i}^j)^{2}}\times \sum_{j=1}^{m}(kw_{S}^j)^{2}
\notag
\end{align}
As Equation \ref{eq:2} now can be represented by an arithmetic circuit containing only multiplication gates and addition gates, we use R1CS to describe each gate in the arithmetic circuit, and further generate the proof circuit for Groth16. Particularly, in the iteration $t$, each client $C_i$ downloads the proof parameters $\sigma _t=\{w_S^t,\tau _c,\tau _e\}$ from the server, and defines the primary input $I_p=\{\sigma _t\} $ and auxiliary input $I_a=\{ w_i^t\} $ used in the zero-knowledge proof. Each client then generates a proof using $I_p$ and $I_a$ as input, and the server verifies the proof using $I_p$.

\subsection{Privacy Preservation for Local Models {via  HE}}
Our privacy protection mechanism is based on homomorphic encryption. 
To protect privacy, all clients hold a same 
random vector  
to mask the true local models. 
Here, we propose to add this random vector to all local models. Hence, the global model update in FedAvg is rewritten as:
\begin{align}
\bar{w}_g^{t+1} \gets r+\frac{1}{n}\sum_{i=1}^{n}w_{i}^{t} 
\gets \sum_{i=1}^{n}  \frac{1}{n} ({w}_{i}^{t} + r)
\gets \sum_{i=1}^{n} \frac{1}{n} \bar{w}_{i}^{t},
\label{eqn:randommask}
\end{align}
where $\bar{w}$ is the value of $w$ after being masked; $r$ is the random vector generated by the Mask Vector Negotiation Protocol (MVNP), which leverages the Paillier homomorphic encryption technique~\cite{paillier1999public}. Algorithm \ref{alg:1} 
shows the details 
on  
generating 
$r$.  
In a word, MVNP ensures all clients obtain the same random vector that is secret to the server, i.e., the server cannot infer any private information from the masked local models as well as their aggregation.  
\par

\begin{algorithm}[!t] 
\caption{Mask Vector Negotiation Protocol (MVNP) via Paillier homomorphic encryption} 
\label{alg:1} 
\footnotesize
\textbf{Input}:A set of $n$ clients; a Paillier encryption public key $\textit{\textbf{pk}}$ and a private key $\textit{\textbf{sk}}$;\\ 
\textbf{Output}: A random mask vector $r$
\begin{algorithmic}[1]
\STATE $//$ Client side.
\FOR{$i=C_1,C_2,\cdots,C_n$}
\STATE $C_i$ generate a random seed $s_i$ randomly.
\STATE Encrypt the $s_i$ with $\textit{\textbf{pk}}$: $[s_i]=Enc(s_i,\textit{\textbf{pk}})$.
\STATE Send $[s_i]$ to the server.
\ENDFOR
\STATE $//$ Server side.
\STATE Compute $[s]=\sum_{i=1}^{n}[s_i]$.
\STATE Broadcast $[s]$.
\STATE $//$ Client side.
\STATE Decrypt $[s]$: $s=Dec([s],\textit{\textbf{sk}})$.
\STATE Generate a random mask vector $r$ over the finite field with the seed $s$.
\RETURN $r$
\end{algorithmic}
\end{algorithm}

\setlength{\textfloatsep}{2mm}

A possible issue is, since the server receives the masked local models, a malicious client may use the valid model to generate the proof, but submits an invalid model  to the server. For example, a malicious client generates a proof using the correct ${w_i}$ and $r$, but submits a wrong mask model using multiplication operations to the server, i.e.  $\bar{w}_i^t=w_i^t*r$, which causes aggregation fails. Hence, the server needs to verify that the model submitted by the client is consistent with the model used to generate the proof. Specifically,  the server needs to use the masked client model $\bar{w_i}$ to check whether each local model  satisfies $\bar{w}_i^t=w_i^t+r$. 
Therefore, $r$ will also be a parameter for the prover to generate a proof and $\bar{w}_i^t$ will be a primary parameter. So, $I_a$ will be  updated to $I_a=\{ w_i^t, r\}$ and $I_p$ will be updated to $I_p=\{\bar{w}_i^t,\sigma _t\}$.

\noindent \textbf{Remark:} Note that the server can subtract two different local models, e.g., $w_{i+1}^{t}-w_{i}^{t}=\bar{w}_{i+1}^{t}-\bar{w}_{i}^{t}$, to remove the effect of the mask value. But it is limited in our scenario as the server only can get the difference of two different local models and the model weights are still secret to the server. Moreover, to our best knowledge, there exist no methods that can recover the training data  only based on the differences between model updates.

\setlength{\textfloatsep}{+2mm}

\subsection{Avoiding Forgery of Proof {via Hash}} 
\label{3.4}
In each iteration $t$, 
clients download  parameters (including $w_S^t$ in plaintext) from the server to generate the proof.
However, a malicious client may use $w_S^t$ as local model to generate a forged proof to fool the server into passing the valid check. 
Assuming a malicious client $C_i$ holds a masked malicious model  $\tilde{w}_i^{t}$ and has downloaded $w_S^{t}$, it can generate a forged proof by letting $\tilde{\pi} _i^t=Prove(pk,I_p,\tilde{I}_a)$, where $\tilde{I}_a=\{w_S^t,\tilde{r} \}$ and $\tilde{r}$ is well-constructed vector defined as    $\tilde{r}=\tilde{w}_i^{t}-w_S^{t}$. The client $C_i$ submits $\tilde{w}_i^{t}$ with the proof $\tilde{\pi} _i^t$ and server will get $1 \gets Verify(vk,\tilde{\pi}_i^t,I_p)$. 
This is because $w_S^t$ in $\tilde{I}_a=\{w_S^t,\tilde{r} \}$ and $w_S^t$ in $I_p=\{\tilde{w}_i^t,\sigma _t\}$ ($\sigma _t=\{w_S^t,\tau _c,\tau _e\}$) satisfy Equation \ref{eq:2},  
and $\tilde{w}_i^{t}$ in $I_p=\{\tilde{w}_i^t,\sigma _t\}$ and $w_S^{t}$, $\tilde{r}$ in $\tilde{I}_a=\{w_S^t,\tilde{r} \}$ satisfy $\tilde{w}_i^{t}=w_S^{t}+\tilde{r}$. Therefore, this malicious model will be allowed to join the aggregation, which will have a negative impact on the global model.
\par
To avoid this, the server must check whether the random mask  $r$ used by clients to generate the proof is the true vector negotiated by MVNP. We propose to use hash function to solve this problem. 
The client $C_i$ generates the hash value $h_i$ of $r_i$ by letting: $h_i=hash(r_i)$, 
and then sends $h_i$ to server. 
The server maintains a set $H,h_i\in H$ and let $h=M\!O\!D\!E(H) $, where $M\!O\!D\!E(H)$ function finds the most frequent value in the set $H$. Based on clients' hash values and their mode, though all malicious clients (less than  50\%) may refuse to compute the true $h_i$, the true $h_i$'s of more than 50\% honest clients are sufficient for the server to obtain the correct $h$. To ensure the proof is not forged, the server only needs to add a constraint to check whether the hash value of $r_i$ used to generate the proof is correct.
\par

\begin{figure*}[!h]
    \centering
    \includegraphics[width=0.8\textwidth]{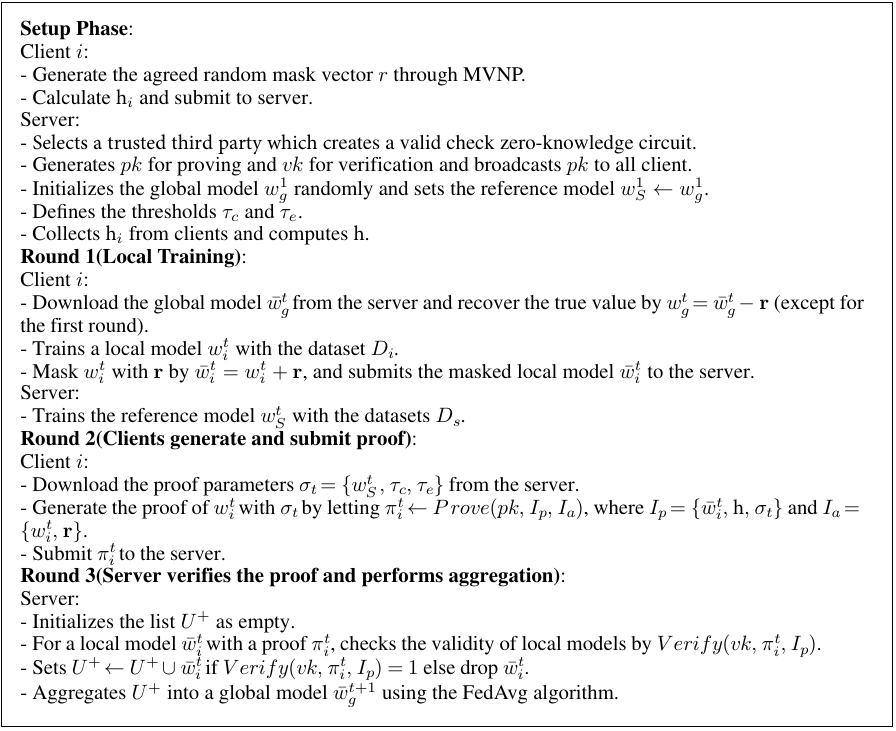}
    \caption{The whole procedure of BPFL.}
    \label{fig:detail_of_bpfl}
\end{figure*}
\subsection{\name Workflow}
The overall workflow of \name 
involves a setup phase followed by three rounds. The whole procedure of BPFL are shown in Figure \ref{fig:detail_of_bpfl}.

\textbf{Setup Phase}. 
The server (1) selects a trusted third party which creates a valid check circuit, generates $pk$ for proving and $vk$ for verification, the server broadcasts $pk$ to all clients; and (2) randomly initializes the global model $w_g^1$ 
and sets the reference model $w_S^1 \gets w_g^1$, and defines the thresholds $\tau_c$ and $\tau_e$. All clients obtain the agreed random mask vector $r$ through MVNP (Algorithm 1), calculate each $h_i$, and send them to the server. The server computes $h$ as the final value. 

\textbf{Round 1 (Local training)}. In iteration $t$, each client $C_i$ first downloads the global model $\bar{w}_g^t$ from the server and recovers the true value by $w_g^t =\bar{w}_g^t-r$ (except for the first round). Then each client $C_i$ trains a local model $w_i^t$ with the dataset $D_i$ using the true global model $w_g^t$, masks $w_i^t$ with $r$, and submits the masked model  $\bar{w}_i^t=w_i^t+r$ to the server. The server trains the reference model $w_S^t$ with the datasets $D_s$, which will be used as one of the parameters for this round to generate a proof.

\textbf{Round 2 (Clients generate and submit proof)}. 
All clients first get the proof parameters $\sigma _t=\{w_S^t,\tau _c,\tau _e\}$ from the server. Then each $C_i$ generates the proof of $w_i^t$ with $\sigma_t$  
by letting $\pi _i^t=Prove(pk,I_p,I_a)$, where $I_p=\{\bar{w}_i^t,h,\sigma _t\}$ 
and $I_a=\{w_i^t, r\}$ and submits $\pi_i^t$ to the server. 

\textbf{Round 3 (Server verifies the proof and performs aggregation)}.  
The validity of the proof will be checked and the well-formed local models will be aggregated by the server $S$. 
Specifically, the server maintains a list, $U^+$ (initialized as empty), 
of local models it has so far identified as valid. The  server checks the validity of all models by verifying the proofs. For a local model $\bar{w}_i^t$ with a proof $\pi_i^t$, if the proof passes the verification, i.e., $Verify(vk,\pi_i^t,I_p)=1$, then the server treats $\bar{w}_i^t$ as a valid model and augments the set 
$U^+ \gets U^+\cup \bar{w}_i^t$. Otherwise, the local model will be flagged as malicious and dropped. 
Every masked local model $\bar{w}_i^t$ in $U^+$ will be aggregated into a global model $\bar{w}_g^{t+1}$ using the FedAvg algorithm and be downloaded by all clients used for the next iteration. 
\par 
\label{tip:update_global_model}
\par

\par

\begin{table}[!t]
    \centering
    \addtolength{\tabcolsep}{4pt}
    \caption{The computation and communication of BPFL}
    \begin{tabular}{ccc}
    \hline
         & Computation cost & Communication cost \\ \hline
        Server & $O(nd)$ & $O(nd)$ \\ 
        Client & $O(d\log d)$ & $O(d)$ \\ \hline
    \end{tabular}
    \label{table:cost}
\end{table}

\section{Theoretical Analysis}
\label{sec:theory}
\noindent {\bf Complexity analysis.}
We show the complexity of \name w.r.t. \#clients $n$ and  \#model parameters $d$. Table \ref{table:cost} show the results, the details of the analysis are outlined below. 

\noindent \emph{1) Computation cost.} In each iteration, each client's computation cost 
can be split into two parts: (1) creating the masked local model is $O(d)$ in \textbf{Round 1}; (2) generating the zero-knowledge proof in \textbf{Round 2}. The prover's computation complexity of Groth16 is $O(m\log m)$, where $m$ is the number of gates in the circuit. According to our valid check constraints in Equation \ref{eq:2}, Equation \ref{eqn:randommask}, and hash function, the number of gates $m$ increases linearly in $d$, so the complexity is $O(d\log d)$; (3) recovering the true global update is $O(d)$ in \textbf{Round 3}. Thus, the overall computation complexity of each client per iteration is $O(d\log d)$. The computation cost of server is mainly at \textbf{Round 3}, which includes verifying the proof for all clients and computing the final aggregation. The verifier's computation 
of Groth16 is $O(1)$ and the aggregation complexity is $O(nd)$. 
Hence, the total computation complexity of the server per iteration is $O(nd)$. 

\noindent \emph{2) Communication cost.} In each iteration, the client sends the local model, which has a $O(d)$ complexity, and a proof, which has a fixed size and the complexity is $O(1)$, to the server. Thus, the communication complexity for every client is $O(d)$. The server's communication costs include: (1) sending the global model to all clients which has a $O(nd)$ complexity; (2) sending the proof parameters to all clients which have a $O(nd)$ complexity. Hence, the overall communication complexity of the server is $O(nd)$.

\noindent {\bf Security analysis.}
We formally 
show the provable security guarantees of 
\name~in the following theorems. 

\begin{restatable}[]{theorem}{privacy}
\label{thm:privacy}
\name~is privacy-preserving, i.e., the honest-but-curious server can not infer clients’ private data if Paillier homomorphic encryption used in Algorithm \ref{alg:1} is semantically secure and Gorth16 is zero-knowledge.
\end{restatable}

\begin{restatable}[]{theorem}{complete}
\label{thm:complete}
\name~accepts local models of honest clients. 
\end{restatable}

\begin{restatable}[]{theorem}{soundness}
\label{thm:soundness}
All updates accepted by \name are valid with probability $1-negl(\kappa)$, where $negl(\cdot)$ is the negligible function and $\kappa$ is the security parameters. 
\end{restatable}

\noindent\emph{Proof of Theorem \ref{thm:privacy}:}

Before representing the proof, we will need to involve the following definition.
\begin{definition}
{A protocol is privacy-preserving  if, for every efficient real-world adversary $\mathcal{A}$, there exists an efficient ideal-world simulator $\mathcal{S}_{\mathcal{A}}$ such that for every efficient environment $\mathcal{E}$ the output of $\mathcal{E}$ when interacting with the adversary $\mathcal{A}$ in a real-world execution and  when interacting with the simulator $\mathcal{S}_{\mathcal{A}}$ in an ideal-world execution are computationally indistinguishable.}
\end{definition}

\privacy*
\begin{proof}
We use the hybrid argument~\cite{cryptoeprint:2021/088} to prove that our \name  satisfies the security Definition 1. To do this, for every real-world (efficient) adversary $\mathcal{A}$, we construct an ideal-world (efficient) simulator $\mathcal{S}$ whose random variable is distributed exactly as real-world.  
The inputs of simulator $\mathcal{S}$ are $s_i$ and $w_i$ that corresponding to the real world, 
We use a sequence of hybrids, each identified by Hyb$_i$ and denote by Output$_{i}(\mathcal{E})$ the output of $\mathcal{E}$ in Hyb$_{i}$, and by Hyb$_{0}$ the real-world execution.
\begin{itemize}
    \item Hyb$_{1}$. Let Hyb$_{1}$ be the same as Hyb$_{0}$, except the followings: The simulator $\mathcal{S}$ replace $s_i$ with random values and send its encrypted value to $\mathcal{E}$.  Given the Paillier homomorphic encryption is semantically secure, Output$_{1}(\mathcal{E})$ is perfectly indistinguishable from Output$_{0}(\mathcal{E})$. 
    \item Hyb$_{2}$. In this hybrid, all clients get the $s$ from server and generate a random vector $r$. Instead of gettting the $r$'s hash value, the simulator $\mathcal{S}$  replaces $r$ with random values (e.g., 0, with appropriate length), then gets the random vector's hash value, and sends it to $\mathcal{E}$. 
    The collision resistance of hash function guarantees Output$_{2}(\mathcal{E})$ is indistinguishable from Output$_{1}(\mathcal{E})$. 
    \item Hyb$_{3}$. In this hybrid, the  simulator $\mathcal{S}$ changes the behavior of all honest clients. Specifically, for each client $C_i$, a uniformly random number is selected to replace the $w_i$. $\mathcal{S}$ receives $w_i$ from $\mathcal{E}$ and outputs the masked $\bar{w}_i$. Because the ideal-world random variable is distributed exactly as real-world, Output$_{3}(\mathcal{E})$ is perfectly indistinguishable from Output$_{2}(\mathcal{E})$.
    \item Hyb$_{4}$. In this hybrid, the simulator $\mathcal{S}$ uses a uniformly random number to replace the $w_i$ and generate a  proof $\pi_i$, then sends it to $\mathcal{E}$. Because the Groth16 algorithm is zero-knowledge, Output$_{4}(\mathcal{E})$ is perfectly indistinguishable from Output$_{3}(\mathcal{E})$.
\end{itemize}
Therefore, the simulator has already completed the simulation since $\mathcal{S}$ successfully simulates real-world without knowing anything about the private inputs. The final hybrid is distributed identically to the operation of $\mathcal{S}$ from the point of view of $\mathcal{E}$. We have thus shown that $\mathcal{E}$'s advantage in distinguishing the interaction with $\mathcal{S}$ from the interaction with $\mathcal{A}$ is negligible.
\end{proof}

\noindent\emph{Proof of Theorem \ref{thm:complete}:}
\complete*
\begin{proof}
The client $C_i$ that$C_i \in C_H$, has a local update $w_i$ that similar to the global update, i.e., satisfies the cosine similarity constraint and the Euclidean distance constraint with the client reference model $w_s$, so that the $w_i$ satisfies Equation \ref{eq:2}. Additionally, $C_i$, by virtue of being honest, will perform the computation using the correct mask value and submit to the server the correctly computed $\bar{w}_i$. The $w_i, r, \bar{w}_i$ also satisfies Equation \ref{eqn:randommask} and hash function. The server will get $Verify(w_i)=1$, i.e., \name accepts the local model of every honest client. This theorem shows the completeness property of \name.
\end{proof}

\noindent\emph{Proof of Theorem \ref{thm:soundness}:}
\soundness*
\begin{proof}
All local models accepted by the server imply $Verify(w_i)=1$. This is true for an honest client according to Theorem \ref{thm:complete}. For a malicious client, it may submit the invalid mask vector hash value to the server during the setup phase, but this does not affect the correct execution of the protocol, because more than half of the honest clients ensure that the server gets the correct mask vector check value. A malicious client may also behave in the following ways:
\begin{enumerate} 
    \item A malicious client masks an invalid local model and submits it to the server. This invalid local model satisfies Equation \ref{eq:2} with probability $negl(\kappa)$ and the server reject the local model with probability $1-negl(\kappa)$.
    \item A malicious client performs an incorrect mask operation on a valid local model and submits it to the server. Since Equation \ref{eqn:randommask} is not satisfied, it will be rejected by the server.
    \item A malicious client uses an incorrect mask vector to perform a masking operation on a valid local model and submits it to the server. The server will check the correctness of the mask vector by the hash value of the mask vector, and since the hash value of the incorrect mask vector is equal to the correct hash value with probability $negl(\kappa)$, the server rejects the local model with probability $1-negl(\kappa)$.
    \item In Round 2, the malicious client attempts to generate a forged proof using the server reference model in the received proof parameters (refer to Section \ref{3.4} for more details). In this process, the malicious client is trying to make the local model satisfy Equation \ref{eqn:randommask}, so it will compute an incorrect random mask vector. The hash value of this incorrect mask vector is equal to the correct hash value with probability $negl(\kappa)$, and the server reject the local model with probability $1-negl(\kappa)$.
\end{enumerate}
So, all updates accepted by \name are valid with probability $1-negl(\kappa)$. This theorem shows the soundness property of \name.
\end{proof}

\begin{figure*}[!t]
    \centering
    \begin{subfigure}{0.24\linewidth}
        \includegraphics[width=1.0\textwidth]{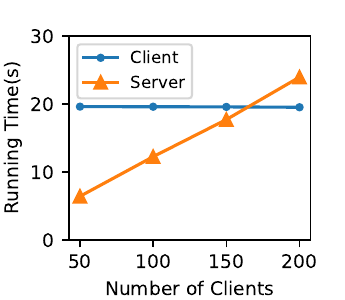}
        \caption{$d=50K$}
        \label{fig:overhead_a}
    \end{subfigure}
    \begin{subfigure}{0.24\linewidth}
        \includegraphics[width=1.0\textwidth]{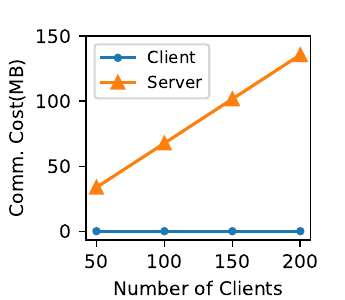}
        \caption{$d=50K$}
        \label{fig:overhead_b}
    \end{subfigure}
    \begin{subfigure}{0.24\linewidth}
        \includegraphics[width=1.0\textwidth]{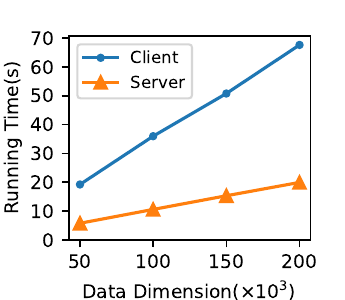}
        \caption{$n=50$}
        \label{fig:overhead_c}
    \end{subfigure}
    \begin{subfigure}{0.24\linewidth}
        \includegraphics[width=1.0\textwidth]{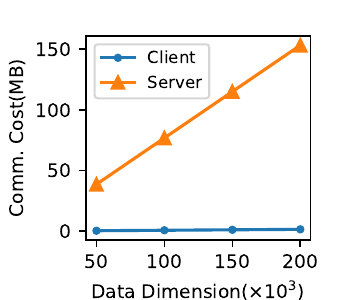}
        \caption{$n=50$}
        \label{fig:overhead_d}
    \end{subfigure}
    \caption{The complete overhead analysis of BPFL. 
    }
    \label{fig:overhead_evaluation}
\end{figure*}
\begin{figure*}[!t]
    \centering
    \begin{subfigure}{0.24\linewidth}
        \includegraphics[width=0.9\textwidth]{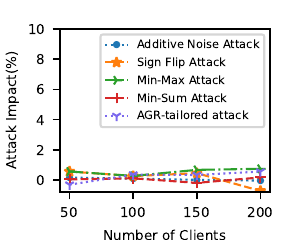}
        \caption{MNIST}
        \label{fig:client_num_a}
    \end{subfigure}
    \begin{subfigure}{0.24\linewidth}
        \includegraphics[width=0.9\textwidth]{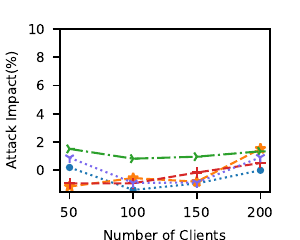}
        \caption{FMNIST}
        \label{fig:client_num_b}
    \end{subfigure}
    \begin{subfigure}{0.24\linewidth}
        \includegraphics[width=0.9\textwidth]{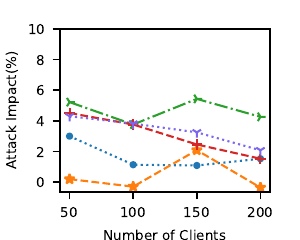}
        \caption{CIFAR10}
        \label{fig:client_num_c}
    \end{subfigure}
    \begin{subfigure}{0.24\linewidth}
        \includegraphics[width=0.9\textwidth]{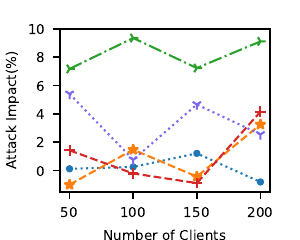}
        \caption{FEMNIST}
        \label{fig:client_num_d}
    \end{subfigure}
    \caption{Attack impact on BPFL under different model poisoning attacks and datasets vs. \#total clients. 
    }
    \label{fig:testing_error_with_client}
    \vspace{-6mm}
\end{figure*}

\begin{figure}[!t]
    \centering
    \begin{subfigure}{0.45\linewidth}
        \includegraphics[width=1.0\textwidth]{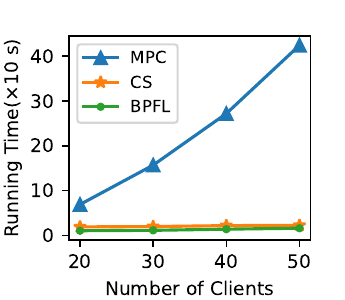}
        \caption{$d=50K$}
        \label{fig:mpc_overhead_a}
    \end{subfigure}
    \begin{subfigure}{0.45\linewidth}
        \includegraphics[width=1.0\textwidth]{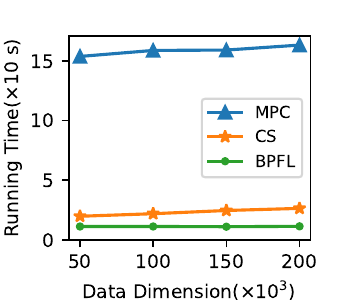}
        \caption{$n=30$}
        \label{fig:mpc_overhead_b}
    \end{subfigure}
    \caption{Comparing overhead of BPFL, CS and MPC.}
    \label{fig:overhead_evaluation_compare_mpc}
\end{figure}

\section{Experiments}
\subsection{Experimental Setup}
We implement \name in Python and use the C++  libsnark  library for zkSNARK proofs. We run experiments on four 
machines, each  with 40 Intel(R) Xeon(R) Silver 4210 CPU at 2.20GHz, 64GB memory, and an NVIDIA 
GTX 1080 Ti GPU. All experiments are under Ubuntu 20.04.4 LTS.
\par

The whole BPFL procedure is illustrated in Figure~\ref{fig:detail_of_bpfl}. 
We evaluate \name  on four image datasets, i.e., \textbf{MNIST}, \textbf{FMNIST}, \textbf{CIFAR10}, and \textbf{FEMNIST}, where the first three datasets are independent identically distributed (IID)  and the last one is non-IID distributed. 

\begin{itemize}
    \item \textbf{MNIST}. A handwritten digit dataset consisting of 60K training images and 10K test images with ten classes. We train a CNN architecture that has five layers and 24,000 parameters for MNIST.
    \item \textbf{FMNIST}. A dataset has a predefined training set of 60K fashion images and a testing set of 10K fashion images. We use LeNet-5 to experiment on FMNIST.
    \item \textbf{CIFAR-10}. A dataset contains RGB images with ten object classes. It has 50K training and 10K test images. We use ResNet-20 and 294,000 parameters for our experiments on CIFAR-10.
    \item \textbf{FEMNIST}. Federated Extended MNIST, built by partitioning the data in Extended MNIST~\cite{cohen2017emnist}, is a non-IID
    dataset with 3,400 clients, 62 classes, and a total of 671,585 grayscale images. We use LeNet-5 to experiment on FEMNIST and we randomly select $n$ out of 3,400 clients for FL training.  
\end{itemize}

\noindent {\bf Parameter setting:} For each dataset, we uniformly select 200 samples from the training set as the server's validation dataset $D_S$ and  
the remaining training data are randomly and evenly divided into $n$ subsets as each client's local training data, where $n$ is  the total number of clients. For the thresholds $\tau_c$ and $\tau_e$, we empirically set $\tau_c=0.99$ on the four datasets; set $\tau_e=0.93$ on MNIST, FMNIST and FEMNIST, and $\tau_e=30.00$ on CIFAR-10, considering the different number of pixels in these datasets. 
We set the clients' local training epochs as 5 and global iterations as 300. 
Our source code is available at the Github: \url{https://github.com/BPFL/BPFL}.

\begin{figure*}[!t]
    \centering
    \begin{subfigure}{0.19\linewidth}
        \includegraphics[width=1.0\textwidth]{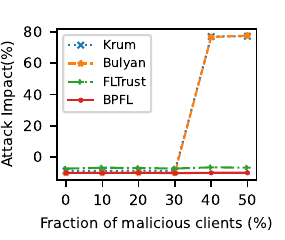}
        \caption{Add Noise}
        \label{fig:mali_num_a_m}
    \end{subfigure}
    \begin{subfigure}{0.19\linewidth}
        \includegraphics[width=1.0\textwidth]{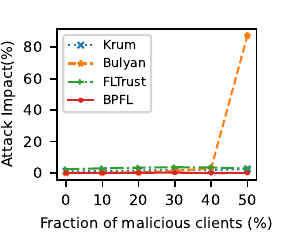}
        \caption{Sign Flip}
        \label{fig:mali_num_b_m}
    \end{subfigure}
    \begin{subfigure}{0.19\linewidth}
        \includegraphics[width=1.0\textwidth]{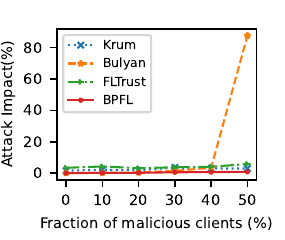}
        \caption{Min-Max}
        \label{fig:mali_num_c_m}
    \end{subfigure}
    \begin{subfigure}{0.19\linewidth}
        \includegraphics[width=1.0\textwidth]{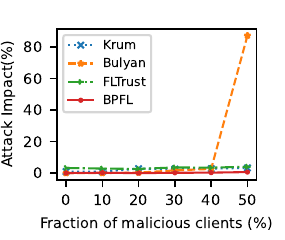}
        \caption{Min-Sum}
        \label{fig:mali_num_d_m}
    \end{subfigure}
    \begin{subfigure}{0.19\linewidth}
        \includegraphics[width=1.0\textwidth]{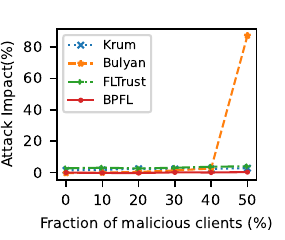}
        \caption{AGR Attack}
        \label{fig:mali_num_e_m}
    \end{subfigure}
    \caption{Attack impact on Byzantine-robust FL methods under different attacks on MNIST vs. malicious clients\%.}
    \label{fig:testing_error_with_malicious_mnist}
\end{figure*}

\begin{figure*}[!t]
    \centering
    \begin{subfigure}{0.19\linewidth}
        \includegraphics[width=1.0\textwidth]{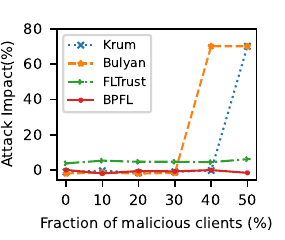}
        \caption{Add Noise}
        \label{fig:mali_num_a_fm}
    \end{subfigure}
    \begin{subfigure}{0.19\linewidth}
        \includegraphics[width=1.0\textwidth]{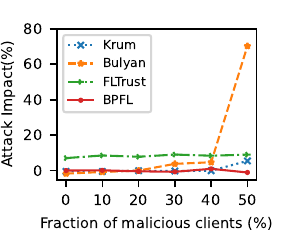}
        \caption{Sign Flip}
        \label{fig:mali_num_b_fm}
    \end{subfigure}
    \begin{subfigure}{0.19\linewidth}
        \includegraphics[width=1.0\textwidth]{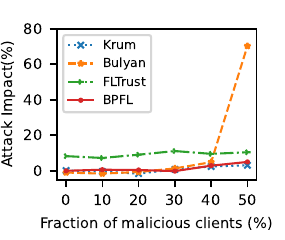}
        \caption{Min-Max}
        \label{fig:mali_num_c_fm}
    \end{subfigure}
    \begin{subfigure}{0.19\linewidth}
        \includegraphics[width=1.0\textwidth]{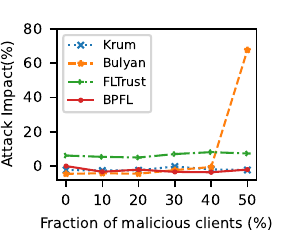}
        \caption{Min-Sum}
        \label{fig:mali_num_d_fm}
    \end{subfigure}
    \begin{subfigure}{0.195\linewidth}
        \includegraphics[width=1.0\textwidth]{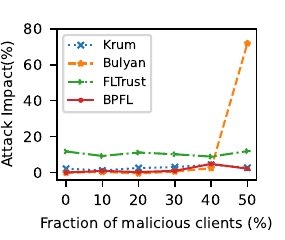}
        \caption{AGR Attack}
        \label{fig:mali_num_e_fm}
    \end{subfigure}
    \caption{Attack impact on Byzantine-robust FL methods under different attacks on FMNIST vs. malicious clients\%.}
    \label{fig:testing_error_with_malicious_fmnist}
\end{figure*}

\begin{figure*}[!t]
    \centering
    \begin{subfigure}{0.19\linewidth}
        \includegraphics[width=1.0\textwidth]{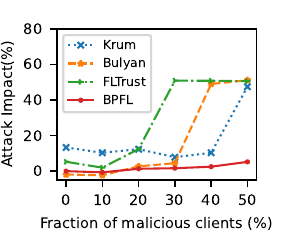}
        \caption{Add Noise}
        \label{fig:mali_num_a_c}
    \end{subfigure}
    \begin{subfigure}{0.19\linewidth}
        \includegraphics[width=1.0\textwidth]{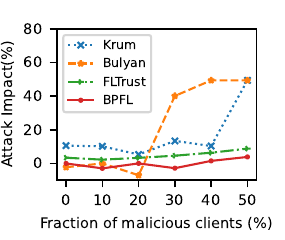}
        \caption{Sign Flip}
        \label{fig:mali_num_b_c}
    \end{subfigure}
    \begin{subfigure}{0.19\linewidth}
        \includegraphics[width=1.0\textwidth]{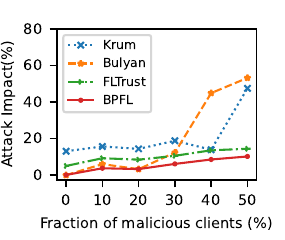}
        \caption{Min-Max}
        \label{fig:mali_num_c_c}
    \end{subfigure}
    \begin{subfigure}{0.19\linewidth}
        \includegraphics[width=1.0\textwidth]{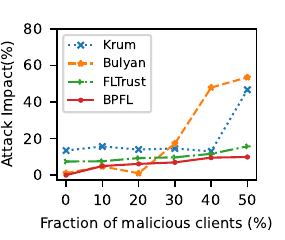}
        \caption{Min-Sum}
        \label{fig:mali_num_d_c}
    \end{subfigure}
    \begin{subfigure}{0.19\linewidth}
        \includegraphics[width=1.0\textwidth]{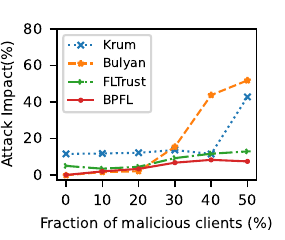}
        \caption{AGR Attack}
        \label{fig:mali_num_e_c}
    \end{subfigure}
    \caption{Attack impact on Byzantine-robust FL methods under different attacks on CIFAR10 vs. malicious clients\%.}
    \label{fig:testing_error_with_malicious_cifar}
\end{figure*}

\begin{figure*}[!t]
    \centering
    \begin{subfigure}{0.19\linewidth}
        \includegraphics[width=1.0\textwidth]{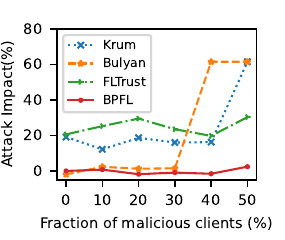}
        \caption{Add Noise}
        \label{fig:mali_num_a_fe}
    \end{subfigure}
    \begin{subfigure}{0.19\linewidth}
        \includegraphics[width=1.0\textwidth]{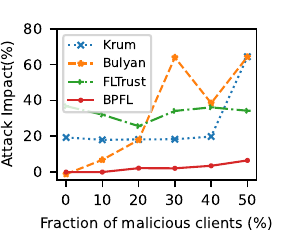}
        \caption{Sign Flip}
        \label{fig:mali_num_b_fe}
    \end{subfigure}
    \begin{subfigure}{0.19\linewidth}
        \includegraphics[width=1.0\textwidth]{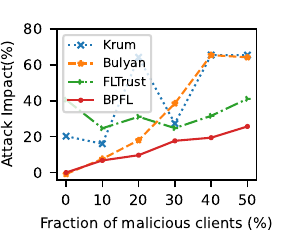}
        \caption{Min-Max}
        \label{fig:mali_num_c_fe}
    \end{subfigure}
    \begin{subfigure}{0.19\linewidth}
        \includegraphics[width=1.0\textwidth]{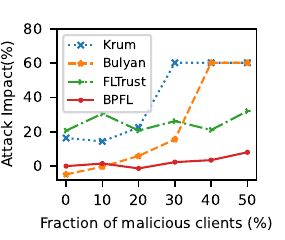}
        \caption{Min-Sum}
        \label{fig:mali_num_d_fe}
    \end{subfigure}
    \begin{subfigure}{0.19\linewidth}
        \includegraphics[width=1.0\textwidth]{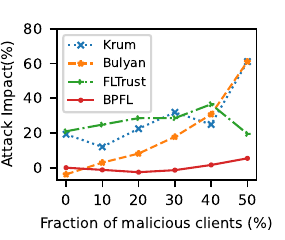}
        \caption{AGR Attack}
        \label{fig:mali_num_e_fe}
    \end{subfigure}
    \caption{Attack impact on Byzantine-robust FL methods under different attacks on FEMNIST vs. malicious clients\%.}
    \label{fig:testing_error_with_malicious_femnist}
\end{figure*}

\subsection{Experimental Results}

\noindent {\bf Overhead evaluation.} 
We first evaluate BPFL's overall overhead and the results are shown in Figure~\ref{fig:overhead_evaluation}. 
In Figure~\ref{fig:overhead_a} and \ref{fig:overhead_b}, we fix the data dimension $d=50K$. We observe the server's runtime and communication cost is  linear in the number of clients $n$ while those 
of clients are almost stable, i.e. independent of $n$. The result is the same as we expected because ZKP enables each client to independently generate a  proof locally without interacting with others and the server also independently verifies each client's proof locally. This ensures \name  will not induce a sharp overhead increase when $n$ 
increases. Figure~\ref{fig:overhead_c} and \ref{fig:overhead_d} show the overhead w.r.t. data dimension $d$, where 
$n=50$.  We see per clients' runtime and communication cost increases almost linearly with $d$ (recall the clients have an $O(d\log d)$ computation complexity and $O(d)$ communication complexity). 
For server, the runtime and communication cost are both linear to $d$. 
This is because both the server's  computation and  communication complexity are $O(d)$. 

Next, we compare BPFL's overhead with the \emph{de facto} MPC-based aggregation method \cite{bonawitz2017practical} and latest Byzantine-robust and privacy-preserving aggregation method CS~\cite{10250809}.  
As MPC is only for privacy-preserving, for fair comparison, we only consider the runtime for  privacy-preserving. 
{
Figure~\ref{fig:overhead_evaluation_compare_mpc} shows the results. 
First, BPFL has the lowest computation cost w.r.t. $n$.  
 This is because  MPC has an $O(n^2 +nd)$ computation complexity and CS has an $O(dk)$ computation complexity, where $k$ is the compressed dimensionality--A larger $k$ yields larger computation time and we set the smallest possible $k$ in our experiments. 
 In contrast, 
 our \name only needs one add operation.} 
Second, the runtime of MPC and SC both increase with the data dimension $d$, while BPFL is insensitive to it. 

Lastly, we evaluate the overhead of each phase of the ZKP for robustness. On the used LeNet-5 network, the setup time is 64.56s, the proof generation time is 28.2s, and the verification time is 0.24s. Indeed, the setup phase of the ZKP is relatively time-consuming, but it 
is executed only once before the iteration and is not affected by the number of clients, hence the run-time is acceptable.

\noindent {\bf Robustness evaluation.}
In this experiment, we evaluate \name in terms of  Byzantine-robustness. For comparison, we also choose three well known Byzantine-robust methods, i.e., Krum~\cite{blanchard2017machine}, Bulyan~\cite{guerraoui2018hidden}, and the state-of-the-art FLTrust~\cite{cao2020fltrust}. 
We consider five model poisoning attacks: (1) \textit{Add Noise Attack}~\cite{li2019rsa}: malicious clients add a noise to the local model; (2) \textit{Sign Flip Attack}~\cite{damaskinos2018asynchronous}: malicious clients flip the sign of their local model; 
(3) \textit{Min-Max Attack}; (4) \textit{Min-Sum Attack}; and (5) \textit{AGR attacks}: three state-of-the-art model poisoning attacks proposed in \cite{shejwalkar2021manipulating}. 
Following~\cite{shejwalkar2021manipulating}, we use the metric \textit{attack impact}, which is defined as the reduction in the accuracy of the global model due to the attack, to measure the impact of attacks. An attack having a large attack impact implies it is more effective.  

\emph{1) Impact of the number of clients:} Figure~\ref{fig:testing_error_with_client} shows the attack impact of the considered five attacks to BPFL, 
where the fraction of malicious clients is 
20\%. 
 We observe that 1) Min-Max, Min-Sum, and AGR attacks have larger attack impacts than Add Noise and Sign Flip attacks, showing they are more effective;  
2) \name achieves a small attack impact 
on the three IID datasets 
under different attacks when facing different number of clients, i.e., less than 6\% (most less than 2\%) in  almost all cases. The attack impact on the non-IID FEMNIST is slightly larger than that on the IID datasets. One reason is that the local models across different clients can be more diverse when trained on non-IID data, which thus makes it more challenging to use a threshold to differentiate between honest models 
and malicious models.

 \emph{2) Impact of the fraction of malicious clients:} 
Figure~\ref{fig:testing_error_with_malicious_mnist} to Figure~\ref{fig:testing_error_with_malicious_femnist} shows the attack impact of the fraction of malicious clients on  Byzantine-robust FL methods against the five attacks, where we fix $n=50$. 
We have the following observations: 1) BPFL performs the best to defend against the five attacks and the attack impact is less than  {$6\%$} in all IID datasets and  less than $10\%$ in the non-IID FEMNIST. This is because BPFL leverages both the Euclidean distance and cosine similarity as the metrics to identify malicious clients, while the remaining robust methods only use one of the two metrics. 2) The state-of-the-art FLTrust performs next to our BPFL. One reason is that both of them train a reference model in the server using a clean validation dataset and use this reference model to guide the correct model update. 
3) Krum and Bulyan fail to defend against some attacks (e.g., Add Noise atack and Sign Flip attack) when the fraction of malicious clients is larger than a certain threshold (e.g., 40\% on MNIST). 
Note that \cite{cao2020fltrust} also draws this conclusion. 
The above observations verify that, it is important to consider two similarity metrics and train a reference model to defend against Byzantine attacks.

 \emph{3) Impact of hyperparametes $\tau_c$ and $\tau_e$:} 
{%
We also test different $\tau_c$ and $\tau_e$,  
e.g.,  $\tau_c=0.95$, $\tau_e=0.9$ on MNIST, FMNIST, and FEMNIST, and $\tau_e=35$ on CIFAR-10. We found the results are almost the same. 
}

\begin{table*}[!t]
\caption{\name and plaintext against model inversion attacks. {$\downarrow$ ($\uparrow$): a smaller (larger) value indicates a larger privacy leakage.}}
\setlength{\tabcolsep}{3mm}
\label{tab:2}
\centering
\addtolength{\tabcolsep}{4pt}
\begin{tabular}{c|c|c|c|c|c|c|c} 
\hline
\multicolumn{2}{l|}{\multirow{2}{*}{}} & \multicolumn{2}{c|}{MSE~$\downarrow$}                                    & \multicolumn{2}{c|}{PSNR~ $\uparrow$} & \multicolumn{2}{c}{SSIM~$\uparrow$}  \\ 
\cline{3-8}
\multicolumn{2}{l|}{}                  & plaintext                                            & BPFL & plaintext & BPFL          & plaintext & BPFL          \\ 
\hline
\multirow{3}{*}{MNIST}   & DLG         & $2.9\times10\textsuperscript{-7}$  & 1.10 & 73.38     & -0.436        & 0.999     & 0             \\ 
\cline{2-8}
                         & iDLG        & $3.8\times10\textsuperscript{-7}$  & 1.11 & 69.201   & -0.469        & 0.989     & 0             \\ 
\cline{2-8}
                         & R-GAP       & $1\times10\textsuperscript{-8}$    & 0.01 & 251.843   & 10.00          & 0.999     & 0.102        \\ 
\hline
\multirow{3}{*}{FMNIST}  & DLG         & $1.9\times10\textsuperscript{-5}$  & 1.22 & 72.668   & -0.839        & 0.999     & 0.001         \\ 
\cline{2-8}
                         & iDLG        & $5.3\times10\textsuperscript{-6}$ & 1.21 & 70.413  & -0.829        & 0.998     & 0.002         \\ 
\cline{2-8}
                         & R-GAP       & $1.2\times10\textsuperscript{-6}$  &0.09& 245.656   & 7.623         & 0.999     & 0.069         \\ 
\hline
\multirow{3}{*}{CIFAR10} & DLG         & $7.9\times10\textsuperscript{-4}$                                                & 1.26 & 50.94    & -0.988        & 0.987     & 0.008         \\ 
\cline{2-8}
                         & iDLG        & $5.8\times10\textsuperscript{-5}$ & 1.26 & 47.297    & -1.01         & 0.998     & 0.008         \\ 
\cline{2-8}
                         & R-GAP       & $4.4\times10\textsuperscript{-5}$  & 0.34 & 31.25    & 4.181         & 0.97      & 0.172         \\
\hline
\multirow{3}{*}{{FEMNIST}} & DLG         & $7\times10\textsuperscript{-8}$                                                & 1.94 & 73.24    & -2.877        & 0.999     & 0.013         \\ 
\cline{2-8}
                         & iDLG        & $1.3\times10\textsuperscript{-7}$ & 1.96 & 71.337    & -2.92         & 0.999     & 0.012         \\ 
\cline{2-8}
                         & R-GAP       & $4.2\times10\textsuperscript{-4}$  & 0.93 & 239.605    & 0.172         & 0.999      & 0       \\
\hline
\end{tabular}
\end{table*}

\noindent {\bf Privacy-preserving evaluation.}
In this experiment, we evaluate three well-known model inversion attacks, i.e.,  DLG~\cite{zhu2019deep}, iDLG \cite{zhao2020idlg}, and R-GAP~\cite{zhu2020r}, against conventional plaintext local models and BPFL's local models on the four datasets.   
To measure the attack effectiveness, we randomly select 100 images in each dataset and compute the average value of Mean Square Error (MSE), Peak Signal to Noise Ratio (PSNR), and Structural Similarity Index Measure (SSIM) between the recovered images by the four attacks and the real images.  
Table \ref{tab:2} shows the results. 
We observe that all attacks achieve good attack performance on plaintext (i.e., low MSE, large PSNR, and large SSIM), but perform poorly on BPFL. We also randomly show in Figure~\ref{fig:rec_BPFL} the recovered images under the DLG attack. 
We observe that for plaintext models, the real images from different datasets can be recovered within 100 iterations. In contrast, with BPFL, the attack cannot recover any useful information of the true images even with 1,000 iterations. This is because BPFL can theoretically protect the client models from being inferred. 

{\bf Detailed image recovery process} The detailed recovery process and recovered images by DLG, iDLG, and R-GAP attacks are shown in 
Figure~\ref{fig:add_pp_evaluation_DLG}, Figure~\ref{fig:add_pp_evaluation_iDLG}, and Figure~\ref{fig:add_pp_evaluation_RGAP}, respectively. 
We can observe that 
the real images from the datasets can be successfully recovered based on plaintext within 100 iterations. While with BPFL, the
attack cannot recover any useful information of the true images, as 
BPFL can theoretically protect the client models from being inferred.
\begin{figure}[htbp]
    \centering
    \includegraphics[width=7cm]{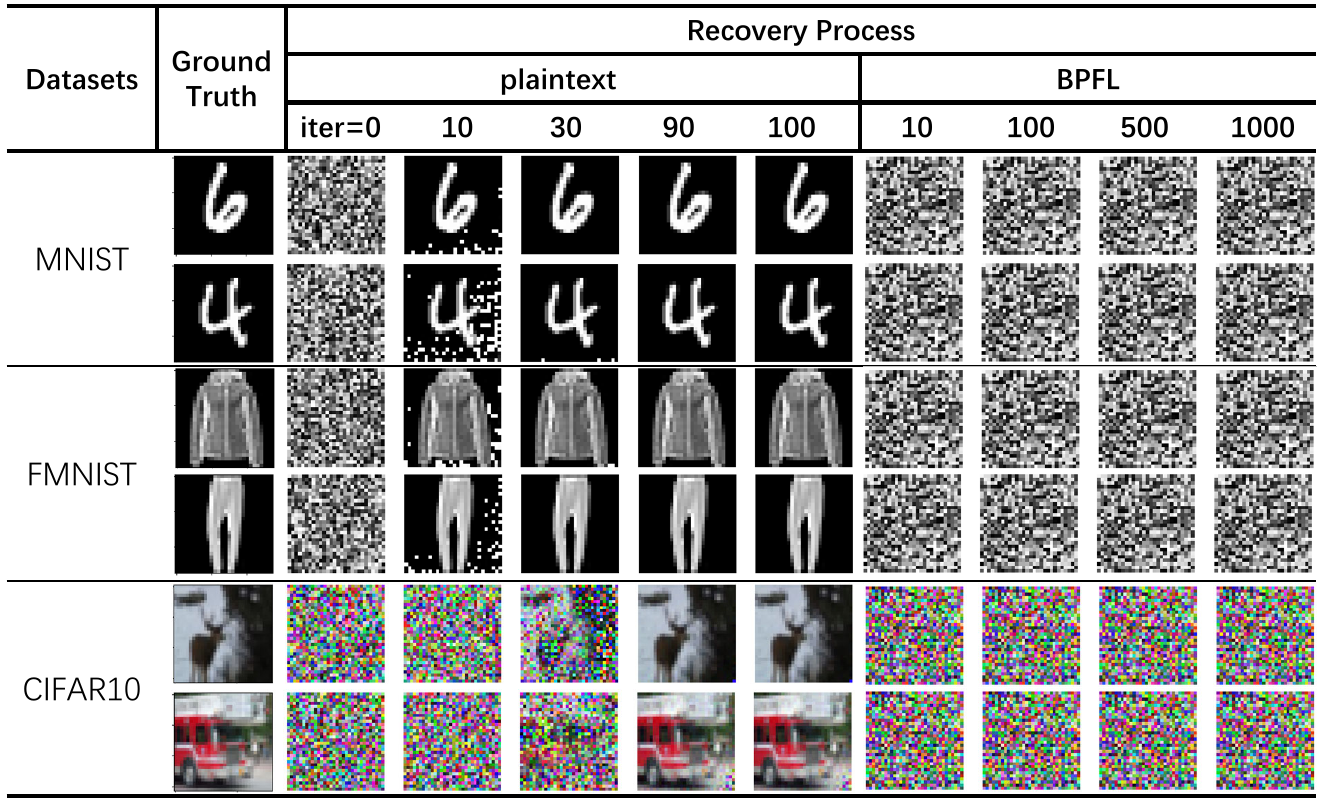}
    \caption{More recovered images by the DLG attack against the plaintext and the privacy-preserving BPFL.}
    \label{fig:add_pp_evaluation_DLG}
\end{figure}

\begin{figure}[htbp]
    \centering
    \includegraphics[width=7cm]{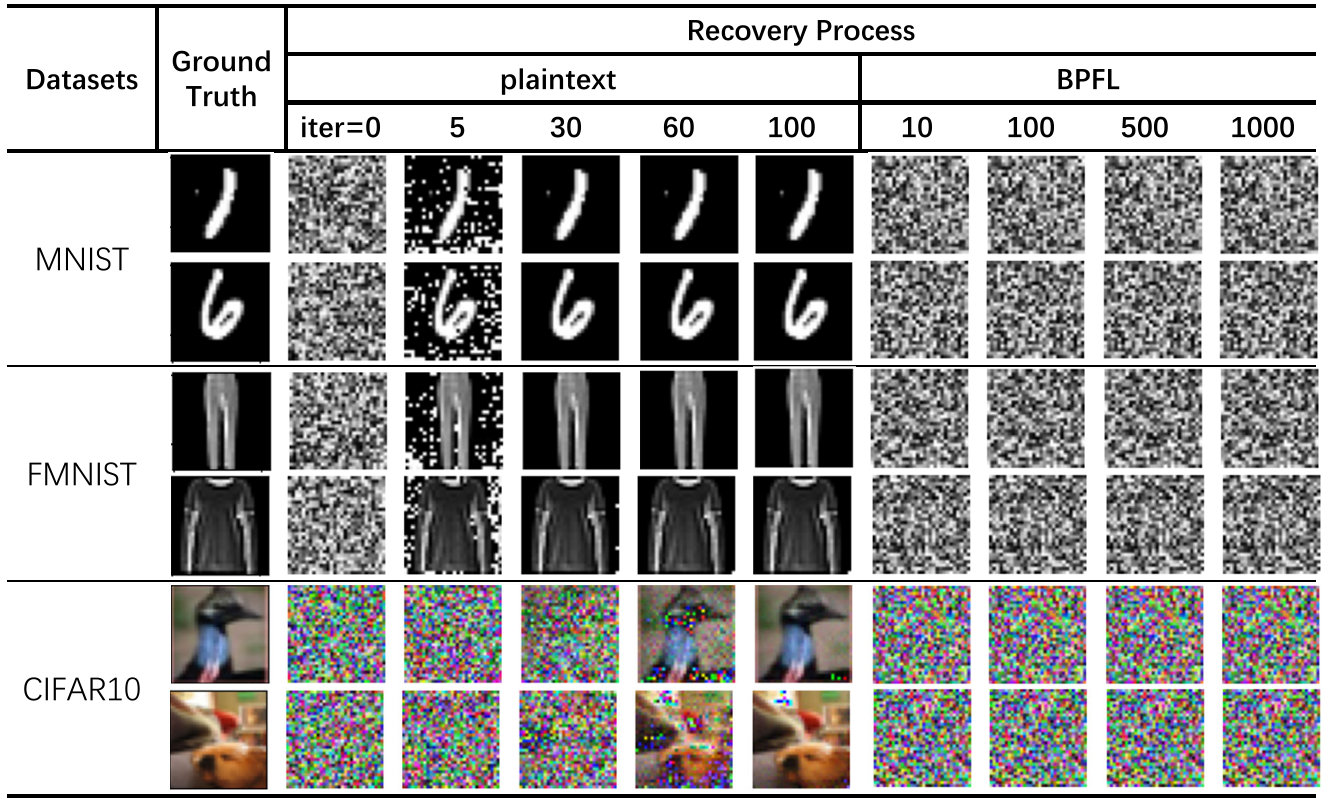}
    \caption{Recovered images by the iDLG attack against the plaintext and the privacy-preserving BPFL.}
    \label{fig:add_pp_evaluation_iDLG}
\end{figure}

\begin{figure}[htbp]
    \centering
    \includegraphics[width=7cm]{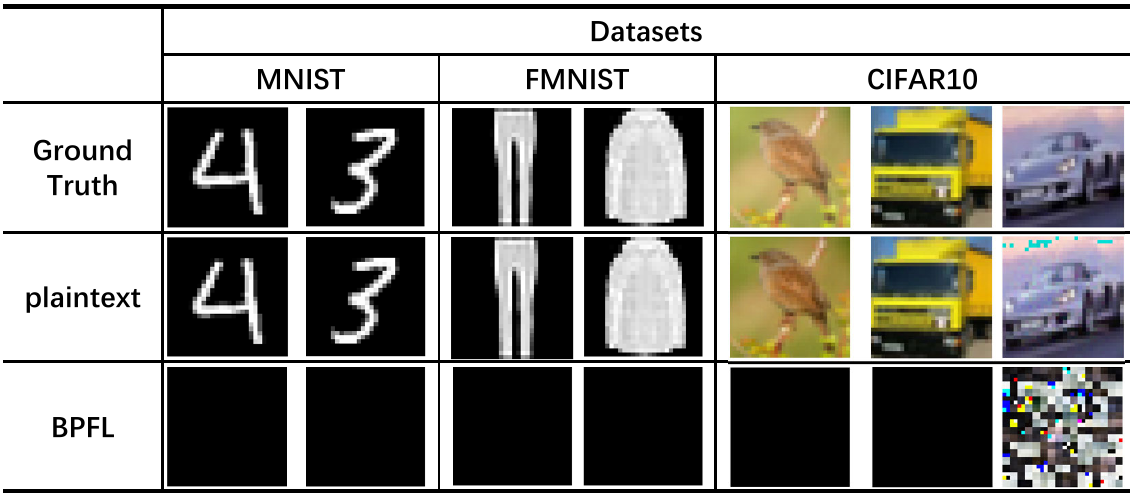}
    \caption{Recovered images by the R-GAP attack against the plaintext and the privacy-preserving BPFL.}
    \label{fig:add_pp_evaluation_RGAP}
    \vspace{-8mm}
\end{figure}

\begin{figure}[htp]
\centering
\begin{minipage}{0.48\textwidth}
\centering
\includegraphics[width=6cm]{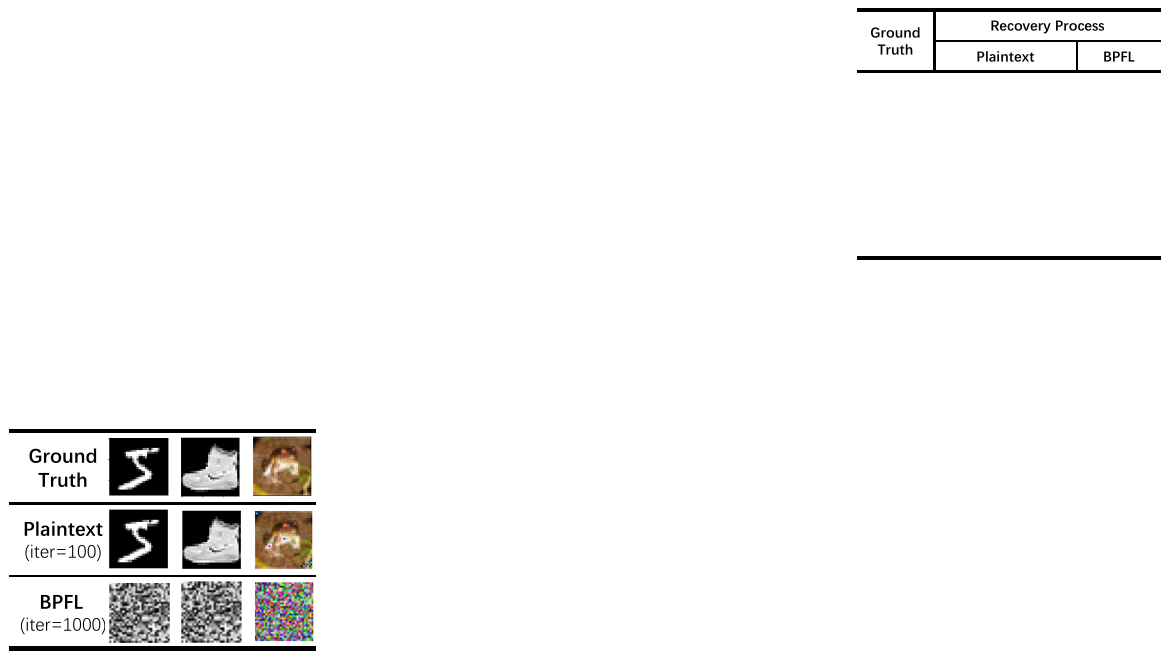}
\caption{Recovered images by the DLG attack to 
    the plaintext and BPFL. 
    }
    \label{fig:rec_BPFL}
\end{minipage}

\begin{minipage}{0.48\textwidth}
\centering
\setlength{\tabcolsep}{6pt}
\addtolength{\tabcolsep}{4pt}
\captionof{table}{Comparing BPFL and representative existing works. B.R.: Byzantine-Robust; P.P.: Privacy-Preserving; 
Eff.: Efficient}
\label{tab:1}
\begin{tabular}{cccc}
\hline
\textbf{Method} & \textbf{B.R.} & \textbf{P.P.} & \textbf{Eff.} \\ \hline
Krum~\cite{blanchard2017machine}&\Checkmark&\XSolidBrush&\Checkmark\\ \hline
\cite{hao2021efficient}&\Checkmark&\Checkmark& \XSolidBrush\\ \hline
\cite{bonawitz2017practical}& \XSolidBrush& \Checkmark & \XSolidBrush \\ \hline
\cite{10093038}  &\Checkmark & \Checkmark & \XSolidBrush  \\ \hline
\cite{wei2020federated}&\XSolidBrush& \Checkmark&\Checkmark\\ \hline
\cite{10250809} & \Checkmark   & \Checkmark  & \XSolidBrush  \\ \hline
FLTrust~\cite{cao2020fltrust}&\Checkmark &\XSolidBrush &\Checkmark\\ \hline
\name   &  \Checkmark   &  \Checkmark   &   \Checkmark \\ \hline
\end{tabular}
\end{minipage}
\end{figure}

\section{Related Work}\label{sec:related}

\noindent {\bf Byzantine-Robust FL.}
A series of Byzantine-robust FL has been proposed recently~\cite{blanchard2017machine,chen2017distributed,guerraoui2018hidden,yin2018byzantine,guerraoui2018hidden,chen2018draco,pillutla2019robust,xie2019zeno,wu2020federated,cao2020fltrust,karimireddy2021learning,farhadkhani2022byzantine}. They majorly leverage the similarity between local client models and the global model to perform robust aggregation.  
\cite{blanchard2017machine} proposed Krum, the first solution to defend against Byzantine attacks. The idea of Krum is to treat a local model as benign if this local model is similar to other local client models, 
where the similarity is measured by Euclidean distance. Krum is shown to  tolerate $f$ malicious clients out of the total $n$ clients, where $n$ and $f$ satisfy $2f+2<n$.  
\cite{yin2018byzantine} proposed a median-based aggregation, where the server obtains the $j$-th parameter of its global model by calculating the median of the $j$-th parameter of all the $n$ local client models. \cite{guerraoui2018hidden} proposed Bulyan, which  combines the idea of Krum and median. Bulyan first iteratively uses Krum to select $k(\le n-2f)$ client model. Then, Bulyan aggregates the $k$ client models using a variant of the trimmed mean. These methods can defend against malicious clients to some extent, but are still not effective enough.  
To further enhance the performance, 
\cite{cao2020fltrust} proposed FLTrust. 
In FLTrust, the server  holds a clean validation dataset, which is used to train a reference model that guides the correct model update direction. 
Specifically, the server assigns a trust score to each client model update based on the cosine similarity between the client model and the reference model. The trust score is 
used as the weight of the client's local update to participate in the model aggregation.
Only client models with relatively larger weights are allowed to participate in the global model aggregation.  
FLTrust is shown to obtain the state-of-the-art robustness.

\noindent {\bf {Privacy-Preserving FL.}}
Existing privacy-preserving FL methods can be mainly divided into three categories: differential privacy (DP)~\cite{pathak2010multiparty,shokri2015privacy,hamm2016learning,mcmahan2018learning,geyer2017differentially,wei2020federated}, secure multi-party computation (MPC)~\cite{danner2015fully,mohassel2017secureml,bonawitz2017practical,melis2019exploiting}, and homomorphic encryption (HE)~\cite{aono2017privacy,zhang2020batchcrypt}. For example, \cite{wei2020federated} proposes a novel DP FL framework, in which carefully designed noises are added to the local client models before aggregation. However, the current
DP-based methods have high utility losses. 
MPC-based methods ensure local clients and the server to jointly complete the aggregation without disclosing the clients' private data. However, they incur an intolerable computation and communication overhead. 
{For instance, the secure aggregation based on MPC in \cite{bonawitz2017practical} has a computation complexity $O(n^2+nd+nd^2)$, which is quadatic in the number of clients $n$ and parameters $d$, while our BPFL is linear to both $n$ and $d$. } 
HE-based methods encrypt local client models before submitting them to the server. Due to the property of HE, the server can complete the global model aggregation without the need to perform decryption. 
HE-based methods obtain  state-of-the-art efficiency and BPFL leverages HE to protect client's data.

\noindent {\bf Byzantine-Robust and Privacy-Preserving FL:} Several works~\cite{hao2021efficient,MA2022103561,9757841,9849010,10093038,10250809} address both privacy and robustness issues. 
However, they are either computation or/and communication inefficient, or have large accuracy loss.  
For instance,  
\cite{10093038} uses three-party computation which has a high computation cost. 
\cite{9757841} are based on DP and robust aggregation. As noted, DP inevitably results in a large accuracy loss.   
\cite{10250809} uses robust aggregation similar to FLTrust and protects privacy via MPC.  

Our \name  simultaneously achieves the both goals 
of byzantine-robust and privacy-preserving 
without incurring excessive overhead to the server. 
Table~\ref{tab:1} shows the comparisons between BPFL and representative existing works.

\section{Conclusion}
\label{sec:conclusion}
We study defenses against Byzantine (security) attacks and data reconstruction  
(privacy) attacks to FL.  
To this end, we propose BPFL, the first FL 
that is efficient, Byzantine-robust, and provably privacy-preserving. 
BPFL seamlessly integrates the ideas of existing Byzantine-robust FL methods, zero-knowledge proof, and homomorphic encryption into a framework.  
Experimental results validate our claim.

{
    \small
    \bibliographystyle{IEEEtran}
    \bibliography{reference,ref_wang}
}

subsection*{  }
\setlength\intextsep{0pt} 
\begin{wrapfigure}{l}{25mm}
    \centering
    \includegraphics[width=1in,height=1.25in,clip,keepaspectratio]{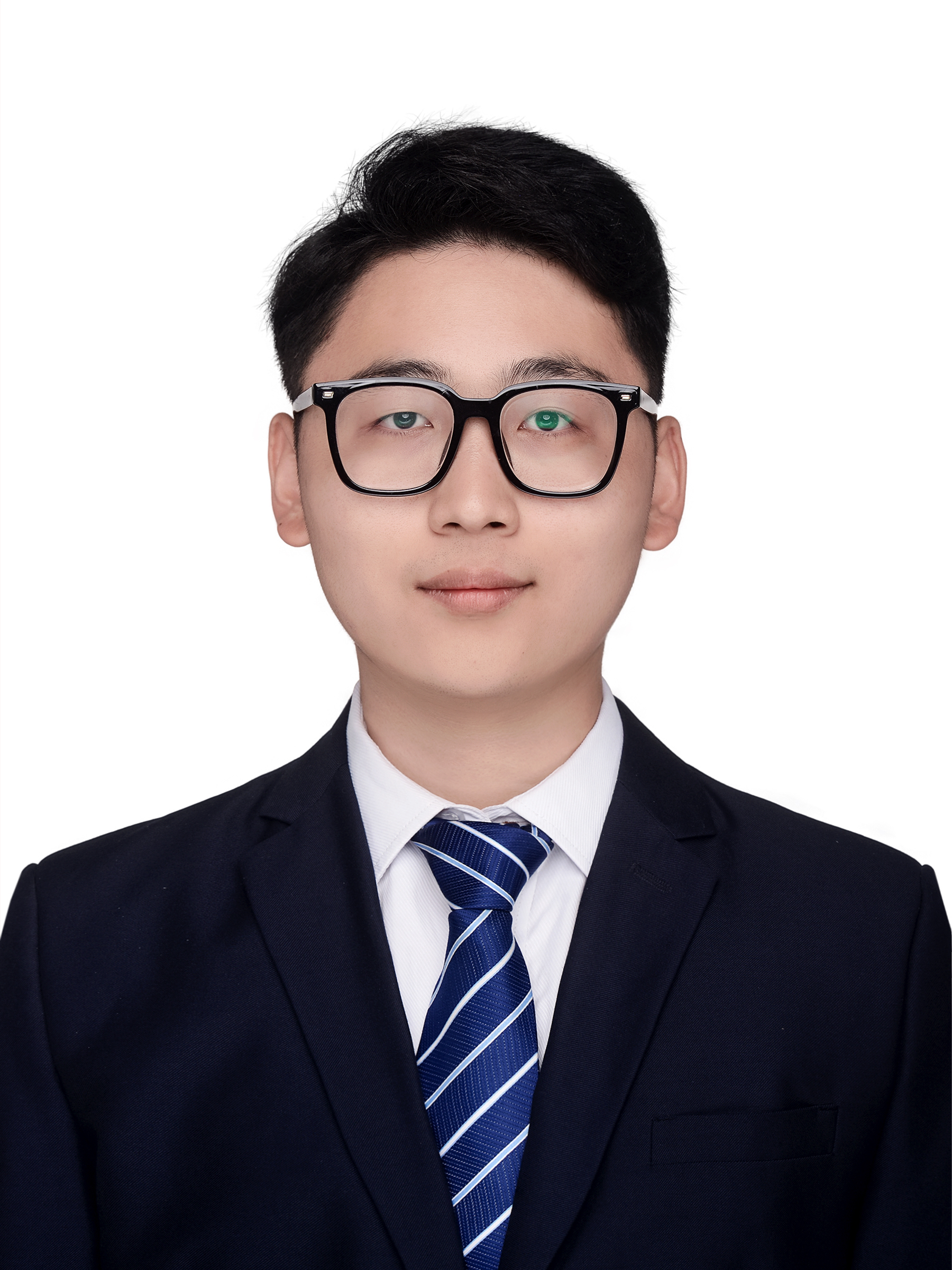}
\end{wrapfigure}
\noindent \textbf{Chenfei Nie} is a graduate student in Computer Science at Jilin University, China. He received his BSc degree also from Jilin University in 2021. His main research interests are gird computing and network security.\par

\hspace*{\fill} 

\subsection*{  } %
\setlength\intextsep{0pt} 
\begin{wrapfigure}{l}{25mm}
    \centering
    \includegraphics[width=1in,height=1.25in,clip,keepaspectratio]{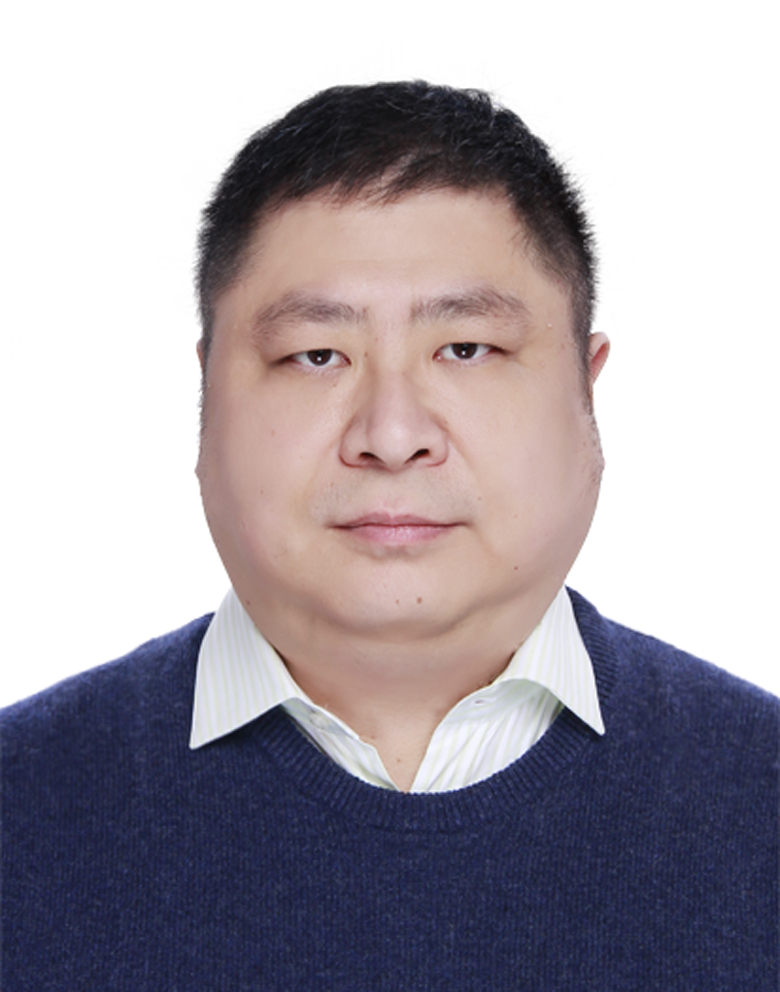}
\end{wrapfigure}
\noindent \textbf{Qiang Li} is currently a professor in Computer Science at Jilin University, China. He received his BSc, MSc and PhD degrees also from Jilin University in 1998, 2001, and 2005, respectively. His main research interests are in artificial intelligence security and privacy preserving.\par

\hspace*{\fill}

\subsection*{  } %
\setlength\intextsep{0pt} 
\begin{wrapfigure}{l}{25mm}
    \centering
    \includegraphics[width=1in,height=1.25in,clip,keepaspectratio]{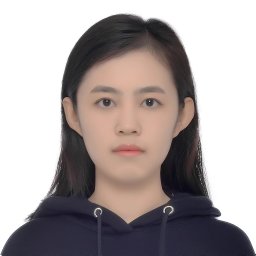}
\end{wrapfigure}
\noindent \textbf{Yuxin Yang} is a Ph.D. student in the College of Computer Science and Technology, Jilin University, China, and Department of Computer Science, Illinois Institute of
Technology, USA. Her main research interests are gird computing and network security.\par

\subsection*{  } %
\setlength\intextsep{0pt} 
\begin{wrapfigure}{l}{25mm}
    \centering
    \includegraphics[width=1in,height=1.25in,clip,keepaspectratio]{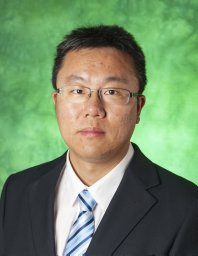}
\end{wrapfigure}
\noindent \textbf{Yuede Ji} received his BEng degree in software engineering from Jilin
University, China in 2012. Currently, he is an MSc candidate in the Computer Science Department at Jilin University, China. His main research focuses on the detection of botnet.\par

\subsection*{  } %
\setlength\intextsep{0pt} 
\begin{wrapfigure}{l}{25mm}
    \centering
    \includegraphics[width=1in,height=1.25in,clip,keepaspectratio]{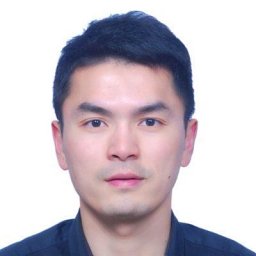}
\end{wrapfigure}
\noindent \textbf{Binghui Wang} received the B.Sc.
degree in network engineering and the M.Sc. degree
in software engineering from the Dalian University of Technology, Dalian, China, in 2012 and
2015, respectively, and the Ph.D. degree in electrical and computer engineering from Iowa State
University, Ames, Iowa, in 2019. His research interests include data-driven
security and privacy, trustworthy machine learning,
and machine learning.\par

\end{document}